\documentclass[11pt,letter]{article}
\usepackage{caption}
\usepackage{subcaption}
\usepackage[colorlinks=true,linkcolor=blue,citecolor=blue]{hyperref}
\usepackage{tikz}
\usepackage{tikz-qtree}
\usepackage{amsmath,amssymb,cleveref}
\usepackage[margin=1in]{geometry}
\usepackage{palatino}
\usepackage{setspace}
\usepackage{amsthm,thmtools, thm-restate}
\usepackage{xcolor} 
\usepackage{comment}
\usepackage{lineno}

\newtheorem{theorem}{Theorem}[section]
\newtheorem{definition}[theorem]{Definition}

\newtheorem{remark}[theorem]{Remark}
\newtheorem{proposition}[theorem]{Proposition}
\newtheorem{lemma}[theorem]{Lemma}
\newtheorem{claim}[theorem]{Claim}
\newtheorem{example}[theorem]{Example}
\newtheorem{corollary}[theorem]{Corollary}

\newcommand{\un}{\mathsf{u}}
\newcommand{\fmux}{\mathsf{MUX}}
\newcommand{\wfmux}{\mathsf{\widehat{MUX}}}
\newcommand{\fmaj}{\mathsf{MAJ}}
\newcommand{\fplu}{\mathsf{P}}
\newcommand{\fand}{\mathsf{AND}}
\newcommand{\fcm}{\mathsf{CM}}
\newcommand{\size}{\mathsf{size}}
\newcommand{\depth}{\mathsf{depth}}
\newcommand{\cc}{\mathsf{cc}}
\newcommand{\sen}{\mathsf{s}}
\newcommand{\bsen}{\mathsf{bs}}

\newcommand{\zuo}{\{0,\un,1\}}
\newcommand{\zuon}{\zuo^n}
\newcommand{\zo}{\{0,1\}}
\newcommand{\zon}{\zo^n}
\newcommand{\bnot}[1]{\overline{#1}}

\newcommand{\hfs}{\sen_u}
\newcommand{\hfdt}{\depth_u}
\newcommand{\hfcc}{\cc_u}
\newcommand{\hfbs}{\bsen_u}

\newcommand{\stabs}{\textrm{\sf stabs}}

\newcommand{\slys}{\textrm{\sf slys}}

\begin{document}

\title{Hazard-free Decision Trees}
\author{Deepu Benson\footnote{CVV Institute of Science and Technology, Ernakulam, bensondeepu@gmail.com}, Balagopal Komarath\footnote{IIT Gandhinagar, bkomarath@rbgo.in}, Jayalal Sarma\footnote{IIT Madras, jayalal@cse.iitm.ac.in}, Nalli Sai Soumya\footnote{Microsoft India, cs20b053@smail.iitm.ac.in}}

\maketitle

\begin{abstract}
Decision trees are one of the most fundamental computational models for computing Boolean functions $f : \{0, 1\}^n \mapsto \{0, 1\}$. It is well-known that the depth and size of decision trees are closely related to time and number of processors respectively for computing functions in the CREW-PRAM model. For a given $f$, a fundamental goal is to minimize the depth and/or the size of the decision tree computing it. 

In this paper, we extend the decision tree model to the world of hazard-free computation. We allow each query to produce three results: zero, one, or unknown. The output could also be: zero, one, or unknown, with the constraint that we should output "unknown" only when we cannot determine the answer from the input bits. This setting naturally gives rise to ternary decision trees computing functions, which we call \emph{hazard-free decision trees}. We prove various lower and upper bounds on the depth and size of hazard-free decision trees and compare them to their Boolean counterparts:
\begin{itemize}
    \item We show that there can be an exponential gap between Boolean decision tree depth (resp. size) denoted $\depth(f)$(resp. $\size(f)$) and hazard-free decision tree depth (resp. size) denoted by $\depth_u(f)$ ($\size_u(f)$) resp.). In particular, the multiplexer function needs only logarithmic depth in the Boolean model but needs full (linear) depth in the hazard-free model. Also, the AND function can be computed in linear size in the Boolean model but requires exponential size in the hazard-free model.
    \item We demonstrate that there are non-trivial functions that are easy in this setting, by constructing an explicit non-degenerate function (family) that can be computed by logarithmic depth hazard-free decision trees. 
    \item We present general methods to construct Hazard-free decision trees from Boolean decision trees. From a Boolean decision tree of size $s$, we construct a hazard-free decision tree of size at most $2^s - 1$ and this is optimal. We also parameterize this result by the maximum number of unknown values in the input. Given a Boolean decision tree of size $s$ and depth $d$, we construct a hazard-free decision tree of size at most $s^{2^{k+1}-1}$ and depth at most $2^k\cdot d$ that will output the correct value as long as the input has at most $k$ unknown values.
    \item On the lower bound side, we show that $\size_u(f) \ge m+M$ where $m$ is the number of prime implicants and $M$ is the number of prime implicates. Using this, we derive explicit functions with close to asymptotically maximum hazard-free decision tree complexity. We also show that $\size_u(f) \geq 2\size(f) - 1$. This lower bound is tight.
\end{itemize}

Viewing the hazard-free decision tree complexity of the hazard-free extension of the function as a complexity measure, we now explore its connection to hazard-free counterparts of other well-studied structural measures in the Boolean world. For example, in the Boolean world, the celebrated sensitivity theorem shows that the sensitivity of a function, a structural parameter, is polynomially equivalent to the (Boolean) decision tree depth.

\begin{itemize}
    \item  We show that the analogues of sensitivity, block sensitivity, and certificate complexity for hazard-free functions are all polynomially equivalent to each other and to hazard-free decision tree depth. i.e., we prove the sensitivity theorem in the hazard-free model. 
    \item We then prove that hazard-free sensitivity satisfies an interesting structural property that is known to hold in the Boolean world. Hazard-free functions with small hazard-free sensitivity are completely determined by their values in any Hamming ball of small radius in $\{0, \un, 1\}^n$. 
\end{itemize}

\end{abstract}

\tableofcontents

\section{Introduction}

A \emph{decision tree} models computation in the following setting. There are two parties: the \emph{querier} and the \emph{responder}. They share a Boolean function $f : \{0, 1\}^n \mapsto  \{0, 1\}$. There is an $x\in \{0, 1\}^n$ that is only known to the responder. The  querier's goal is to determine $f(x)$ by issuing queries $x_i$, asking for the $i^\text{th}$ bit of the input, to the responder. The \emph{worst-case time} taken by the querier is the maximum number of queries required to find the answer, where the maximum is taken over all $x$. 
The querier's aim is to minimize the number of queries required to determine the value of the function, in all cases, even if the function depends on all its inputs.

\begin{example}
    Consider the function $\fmux_1(s_0, x_0, x_1) := \bnot{s_0}x_0 + s_0x_1$ \footnote{Following standard notation in electronic circuit design, we use $+$ to denote Boolean OR, juxtaposition or multiplication for AND, and overline for negation.}. This function clearly depends on all three inputs. However, there is a strategy where the time taken is at most two. First, the querier asks for the value of $s_0$. If $s_0 = 0$, then the querier asks for $x_0$ and the answer to this query is the output. If $s_0 = 1$, then the querier asks for $x_1$ and the answer to this query is the output.
\end{example}

A \emph{decision tree} represents the strategy of the querier using a tree structure. The root of the tree represents the first query. The downward (upward) arrow is the decision tree followed by the querier if the reply to the first query is $0$ ($1$, resp.). Finally, the leaves of the tree are labeled with $0$ or $1$ and represent the answer determined by the querier. The \emph{depth} of the decision tree represents the worst-case time.

\begin{figure}[h]
\centering
\begin{subfigure}[b]{0.35\textwidth}
\begin{tikzpicture}[grow=right]
  \tikzset{level distance=1.25cm}
  \Tree [.{$s_0$} [.{$x_0$} [.0 ] [.1 ] ] [.{$x_1$} [.0 ] [.1 ] ] ]
\end{tikzpicture}
\caption{Boolean decision tree for $\fmux_1$\label{fig:mux1}}
\end{subfigure}
\begin{subfigure}[b]{0.5\textwidth}
\begin{tikzpicture}[grow=right]
    \tikzset{level distance=2cm}
    \Tree [.{$s_0$}
      [.{$x_0$} [.0 ] [.{$\un$} ] [.1 ] ]
      [.{$x_0$} [.{$x_1$} [.0 ] [.{$\un$} ] [.{$\un$} ] ]
           [.{$\un$} ]
           [.{$x_1$} [.{$\un$} ] [.{$\un$} ] [.1 ] ] ]
      [.{$x_1$} [.0 ] [.{$\un$} ] [.1 ] ]
    ]
\end{tikzpicture}
    \caption{Hazard-free decision tree for $\fmux_1$\label{fig:hfmux1}}
\end{subfigure}
\caption{Decision trees for $\fmux_1$}
\end{figure}

In this paper, we consider decision trees in the presence of \emph{uncertainty}. Uncertainty is modeled using Kleene's logic of uncertainty following its use in various computational models for the same purpose \cite{Huff,ikenmeyer2019complexity,IKS}. This algebra is over a three element set $\{ 0, \un, 1 \}$ where $0$ and $1$ are \emph{certain} or \emph{stable} values and $\un$ represents all \emph{uncertain} values and is considered the canonical \emph{unstable} value. Useful computation can happen even in the presence of uncertainty. Consider the function $\fmux_1$ and the sequence of queries and replies: $\langle s_0, \un \rangle, \langle x_0, 1 \rangle, \langle x_1, 1 \rangle$. The querier can determine that the output of the function is $1$ even though the value of one of the queried bits is $\un$. For any Boolean function, there is a unique extension to the domain $\zuon$ that outputs stable values whenever possible. This function is called the hazard-free extension of $f$ (See Definition~\ref{def:hfe}).

\begin{example}\label{ex:hfmux1}
    Consider the hazard-free extension of $\fmux_1$. We will prove that three queries are required by forcing the querier to query all inputs. Suppose the querier asks for $s_0$ first, then we reply $\un$ and now the querier is forced to query both $x_0$ and $x_1$. Suppose the querier asks for $x_0$ first, then we reply $1$. Now, if the second query is $x_1$, we reply $0$ and now the querier is forced to ask for $s_0$. If the second query is $s$, we reply $1$ and force the querier to ask for $x_1$. Similarly, we can argue that querier has to ask for all the bits if the first query is for $x_1$.
\end{example}

\begin{figure}

\end{figure}

We can represent decision trees in the presence of uncertainty using ternary trees. The tree in Figure~\ref{fig:hfmux1} represents the querier strategy given in Example~\ref{ex:hfmux1}. The above example shows that the time taken or the depth of a decision tree may increase in the presence of uncertainty. How worse can it get?

The Boolean function
$$\fmux_n(s_0, s_1, \dotsc, s_{n-1}, (x_{(b_0, \dotsc, b_{n-1})})_{b_i\in\{0,1\}}) = x_{(s_0, \dotsc, s_{n-1})}$$
is a function on $n+2^n$ inputs that depends on all its inputs and has a Boolean decision tree of depth $n+1$ that computes it. The bits $s_i$ are called the \emph{selector} bits and $x_j$ are called \emph{data} bits. Notice that any decision tree that computes a function that depends on all its $N$ inputs must have logarithmic (in $N$) depth. Therefore, $\fmux_n$ is one of the easiest functions to compute in the Boolean decision tree model. Notice that the Boolean decision tree for $\fmux_1$ in Figure~\ref{fig:mux1} has depth two. In contrast, the hazard-free decision tree for $\fmux_1$ in Figure~\ref{fig:hfmux1} has depth three. We generalize this observation by showing that $\fmux_n$ requires full ($2^n+n$) depth in the presence of uncertainty. This shows that the time taken can be exponentially more in the presence of uncertainty.

\begin{restatable}{theorem}{muxdepth}
    \label{thm:muxdepth}
    $\depth_\un(\fmux_n) = 2^n+n$
\end{restatable}

\paragraph{Proof idea. } The depth lower bound is proved using a standard adversarial argument that extends the argument in Example~\ref{ex:hfmux1}. \qed

Nisan showed that functions with decision tree of depth $d$ can be evaluated by CREW-PRAM\footnote{Consecutive Read Exclusive Write Parallel RAM (CREW-PRAM) model is a collection of synchronized processors computing in parallel with access to a shared memory with no write conflicts.} in $\log(d)$ time. The number of processors required by this algorithm is at most the \emph{size} of the decision tree. The size of a decision tree is the number of leaves in the tree\footnote{Since decision trees are complete binary trees, the number of leaves is always exactly one more than the number of internal nodes. Therefore, this measure is at most a factor of two less than the total number of nodes.}. Observe that the Boolean decision tree for $\fmux_1$ in Figure~\ref{fig:mux1} only has four leaves whereas the tree in Figure~\ref{fig:hfmux1} has thirteen leaves. We prove that a quadratic blowup in size is necessary in general.

\begin{restatable}{theorem}{muxsize}
    \label{thm:muxsize}
    $2\cdot 4^n \leq \size_\un(\fmux_n) \leq 4^{n+1} - 3^n$.
\end{restatable}
\paragraph{Proof idea. } The size upper bound is obtained by a straight-forward generalization of the tree in Figure~\ref{fig:hfmux1}. The proof of the size lower bound is trickier and involves carefully crafting a set of hard inputs that must reach distinct leaves in the tree. \qed

We also present a general construction of hazard-free decision trees from Boolean decision trees. In fact, our upper bound construction for $\fmux_n$ can be viewed as an optimization of the upper bound construction in the following theorem.
\begin{restatable}{theorem}{dthfdt}
    \label{thm:dthfdt}
    $2\size(f)-1 \leq \size_\un(f) \leq 2^{\size(f)} - 1$.
\end{restatable}
\paragraph{Proof idea. } The lower bound is easy. The upper bound is proved by recursively constructing a hazard-free decision tree from a Boolean decision tree. We construct the $\un$ subtree of the root node in the hazard-free decision tree from the $0$ and $1$ subtrees that have been constructed recursively. The $\un$ subtree construction can be viewed as a product construction where we replace each leaf in a copy of the $0$ subtree with a modified copy of the $1$ subtree. This product construction works because for every input that reaches the $\un$ subtree, the output value can determined by the output values when that bit is $0$ and when it is $1$, which is exactly what the $0$ and $1$ subtrees compute. \qed

The lower bound in this proof is tight as witnessed by the parity function on $n$ inputs.
Notice that the size blowup when constructing hazard-free decision trees using the above procedure can be exponential. We can show that this is unavoidable by considering the $\fand_n$ function. This function has a Boolean decision tree of size $n+1$.
\begin{restatable}{theorem}{andsize}
    \label{thm:andsize}
    $\size_\un(\fand_n) = 2^{n+1}-1$.
\end{restatable}

We can do better than the construction in Theorem~\ref{thm:dthfdt} if we know that the number of unstable values in the input is bounded. Define $k$-bit hazard-free decision tree as the decision tree which hazard-free for input settings with at most $k$ number of $u$'s in them. The following is an analogue of Corollary~{1.10} by Ikenmeyer et.~al.~\cite{ikenmeyer2019complexity}.

\begin{restatable}{theorem}{kbitdthfdt}
    Let $T$ be a Boolean decision tree of size $s$ and depth $d$ for $f$. Then, there exists a $k$-bit hazard-free decision tree of size at most $s^{2^{k+1} - 1}$ and depth at most $2^k \cdot d$ for $f$.
\end{restatable}

Note that due to the exponential (in $k$) blowup of the exponent, this construction is mainly interesting for small values of $k$. If $k$ is a constant, then the size blowup is only polynomial, and the depth blows up only by a constant factor.

The hardness of $\fmux_n$ and $\fand_n$ begs the question: are there functions that are easy in the presence of uncertainty? It is easy to construct such functions if they are degenerate\footnote{A function $f$ \emph{depends on input bit $i$} if there are two inputs $x$ and $y$ that differ only at the $i^\text{th}$ position and $f(x) \neq f(y)$. A function that does not depend on at least one of its input bits is called \emph{degenerate}.}. Consider $f(x_1,\dotsc,x_n) = x_1$. This function has a hazard-free decision tree of depth $1$. We want functions that are non-degenerate.
\begin{restatable}{theorem}{smalldepth}
    \label{thm:smalldepth}
    For any $n \geq 1$, there is a non-degenerate Boolean function on $n+2^{n+1}-1$ inputs that has a hazard-free decision tree of depth $2n+1$.
\end{restatable}
We believe this theorem was necessary to establish that the theory of decision tree depth in the presence of uncertainty is as interesting as Boolean decision tree depth.

\paragraph{Proof idea. } We construct this function by composing $\fmux_n$ with a set of carefully chosen functions. The function $\fmux_n(s_0,\dotsc,s_{n-1}, x_{(0,\dotsc,0)}, \dotsc, x_{(1,\dotsc,1)})$ has low Boolean decision tree depth because for any values of $s_i$, $2^n-1$ input bits of the form $x_j$ become insignificant. Therefore, once the $n$ selector bits are queried, only one more bit is relevant. However, in the hazard-free setting, the selector bits could have value $\un$. In the worst-case, if all selector bits are $\un$, we will be forced to query all the remaining bits because all of them are \emph{independent} bits. We could make all of them dependent by considering $\fmux_n(s_0,\dotsc,s_{n-1}, y, \dotsc, y)$ where $y$ is a new input. But this function only depends on $y$. In other words, we have made them too dependent. We build our easy, non-degenerate function by composing $\fmux_n$ with a set of functions that are dependent enough so that we can determine the output with a small number of queries, for any values of selector bits, but on the other hand, are independent enough so that the function depends on all the inputs. \qed

Any function on $N$ inputs has a decision tree of size at most $2^N$. However, the hazard-free decision tree size could be as big as $3^N$ as it is a ternary tree of depth at most $N$. However, it is not immediately apparent (using a counting argument, for example) whether there are functions that need $3^N$ size. This is because there are only $2^{2^N}$ hazard-free extensions of Boolean functions, which is much smaller than the number of ternary functions, which is $3^{3^N}$. The following theorem helps us to lower bound the size of hazard-free decision trees (See Definition~\ref{def:pi} for the definition of prime implicants and implicates).
\begin{restatable}{theorem}{imp}
    \label{thm:imp}
    Let $f$ be any function. Then, $\size_\un(f) \geq m+M$ where $m$ is the number of prime implicants of $f$ and $M$ is the number of prime implicates of $f$.
\end{restatable}
Using Theorem~\ref{thm:imp}, we show that there are functions on $N$ inputs that require hazard-free decision trees of size close to $3^N$, the maximum possible.
\begin{corollary}
    Let $N = 3n$ for $n \geq 1$. Let $\fcm_N$ be the $N$-input Chandra-Markowsky function \cite{CM}. Then, $\size_\un(\fcm_N) \geq \binom{N}{N/3}\binom{2N/3}{N/3} \sim 3^N/N^2$.
\end{corollary}
Note that the size lower bound given in Theorem~\ref{thm:muxsize} cannot be obtained by applying Theorem~\ref{thm:imp} because $\fmux_n$ has only $3^n$ prime implicants and $3^n$ prime implicates. Theorem~\ref{thm:imp} only counts the number of leaves labeled $0$ and $1$ in the decision tree. However, as can be observed in the upper bound construction in Theorem~\ref{thm:muxsize}, the number of leaves in the decision tree is dominated by those labeled $\un$.

\paragraph{Function complexity parameters:}
A crucial direction in Boolean function complexity is to relate the computational complexity to structural parameters (parameters that do not refer to a notion of computation, such as CREW-PRAM model) of the Boolean function. In Boolean function complexity theory, various parameters such as sensitivity, block sensitivity, and certificate complexity are studied extensively for this purpose. Sensitivity, denoted $\sen(f)$, in particular, originated from the works of \cite{CD82,CDR86,R82}, where they showed a lower bound of logarithm of sensitivity, for the number of steps required to compute a Boolean function on a CREW-PRAM model. A few years later, Nisan~\cite{Nis89} showed a characterization of the time required to compute a function by CREW-PRAM model, by a variant of the sensitivity called {\em block sensitivity}, denoted $\bsen(f)$. More precisely, it was shown that the time taken by CREW-PRAM model is $\Theta(\log \bsen(f))$. After a series of works spanning over three decades, and culminating with the celebrated proof of Huang~\cite{Hua19} of the sensitivity conjecture, the two parameters $\bsen(f)$ and $\sen(f)$ are polynomially related - that is $\bsen(f) \le s(f)^4$. This finally established that decision tree depth, certificate complexity, block sensitivity, and sensitivity are all polynomially equivalent.

We will define analogues of these parameters for hazard-free functions. We will prove that they are all polynomially equivalent to hazard-free decision tree depth, i.e., the sensitivity theorem holds in the hazard-free world as well.

\emph{Certificate complexity} is a measure of Boolean function complexity in the following setting: There is a prover and a verifier who share knowledge of a function $f:\{0, 1\}^n \mapsto \{0, 1\}$. The prover knows an $x\in \{0, 1\}^n$ and wants to convince the verifier that the output of $f$ on $x$ is $f(x)$ by revealing the \emph{the minimum number of bits} in $x$. For example, let $f$ be $\fmux_n$. Then, for any input, prover can reveal all $n$ selector bits and the data bit indexed by the selector bits. Therefore, the certificate complexity of $\fmux_n$ is at most $n+1$. Decision tree depth is at most the square of certificate complexity and at least the certificate complexity for all Boolean functions. An alternate definition of certificate complexity is that it is the smallest $k$ such that the function can be written as a $k$-CNF \emph{and} a $k$-DNF.

We extend certificate complexity to the hazard-free setting and define \emph{hazard-free certificate complexity}. For $f:\{0, \un, 1\}^n \mapsto \{0, \un, 1\}$ and $x\in \{0, \un, 1\}^n$, it is defined as the minimum number of values in $x$ the prover should reveal to communicate $f(x)$ to the verifier. As a simple example, consider the function $\fmux_2$, its hazard-free certificate complexity is $4$. If the selector bits are not $\un\un$, then the prover can reveal the selector bits and all the indexed data bits. This is at most $4$ values in the input. Otherwise, selector bits are $\un\un$. If the output is a stable value, the prover can reveal all four data bits.  Otherwise, the prover can reveal the selector bits and two data bits that have distinct values or a data bit that is $\un$ to convince the verifier. Recall that the hazard-free decision tree depth of $\fmux_2$ is six. Therefore, even this simple example shows that hazard-free certificate complexity can be strictly smaller than hazard-free decision tree depth. However, just like in the Boolean setting, we establish that hazard-free decision tree depth and hazard-free certificate complexity are polynomially related.

\begin{restatable}{theorem}{dtcc}
    \label{thm:dtcc}
    $\cc_\un(f) \leq \depth_\un(f) \leq 4\cc_\un^3(f)$.
\end{restatable}
\paragraph{Proof idea. } The first inequality is easy to prove. The prover simply reveals values that would be obtained by the querier. The proof of the second inequality is similar to the proof of the analogous theorem in the Boolean world, which is by induction on the number of inputs to the function. However, restrictions of Boolean functions are Boolean. But, restrictions of hazard-free extensions of Boolean functions are not necessarily hazard-free extensions of Boolean functions. We introduce the notion of \emph{weakly hazard-free functions}, which are closed under restrictions and is a larger class of functions than hazard-free functions, in the induction to avoid this roadblock. \qed

There is also a quadratic separation between certificate complexity and decision tree depth in this setting. The same function used to achieve this separation in the Boolean case works (See Chapter~14, Page~408, \cite{BFC}) when combined with Theorem~\ref{thm:sbscc}.

\emph{Sensitivity} of a function is the number of bits that can be changed (individually) to change the output of the function. \emph{Block sensitivity} is similar but allows flipping blocks of bits to change the output of the function. Both parameters are polynomially equivalent to decision tree depth in the Boolean setting. Therefore, to show that a function has small depth, it suffices to show that it has small sensitivity. A natural question is whether or not such relations hold for hazard-free functions as well. We answer this question positively by using hazard-free DNF and CNF formulas as an intermediate step.

\begin{restatable}{theorem}{sbscc}
    \label{thm:sbscc}
    Let $f$ be a function. Let $k_1$ and $k_2$ be such that $f$ has a hazard-free $k_1$-DNF formula and a hazard-free $k_2$-CNF formula. Then, we have:
    \begin{equation*}
        \max\{k_1, k_2\} \leq \sen_\un(f) \leq \bsen_\un(f) \leq \cc_\un(f) \leq k_1+k_2-1.
    \end{equation*}
\end{restatable}

The second and third inequalities are true by definition. The first and last inequalities follow from properties of hazard-free CNF and DNF formulas established by Huffman \cite{Huff}. This proof also shows that the quadratic separation between sensitivity and block sensitivity as shown by \cite{Rubinstein} does not hold in the hazard-free world.

Recall that certificate complexity of Boolean functions is defined in the following manner too: It is the smallest $k$ such that $f$ can be written as a $k$-CNF \emph{and} a $k$-DNF formula (not necessarily hazard-free). Theorem~\ref{thm:sbscc} hints that the certificate complexity for hazard-free functions could be larger than even the width of hazard-free CNF and DNF formulas. That is, it is sometimes strictly harder to prove that the output is $\un$ than proving a $0$ or a $1$ output. Consider the majority function on $2n+1$ input bits. For any input where the output is $0$ or $1$, the prover can reveal exactly $n+1$ values to convince the verifier. Consider the input $0^n1^n\un$ for which the output is $\un$. We claim that the prover must reveal all values to convince the verifier that the output is $\un$. Indeed, changing any $0$ to a $1$ or a $1$ to a $0$ changes the output from $\un$ to a stable value. Therefore, the prover must reveal all stable values in this input. Moreover, changing the $\un$ to a stable value makes the output stable as well. So the prover must reveal the solitary $\un$ in the input as well. Notice that this also shows the last inequality in Theorem~\ref{thm:sbscc} is tight.

We note that the connection between hazard-free computation and prime implicants/prime implicates have been established by Jukna \cite{Jukna21}. In \cref{sec:CombProof}, we establish \cref{thm:sbscc} by using just the structural properties of the hypercube. We use the natural the notion of the subcube representing an input $x \in \zuon$ that contains exactly the resolutions of $x$ in the proof and also extensively use implicant and implicate subcubes. This alternate proof reveals more insights about the structure of inputs at which the maximumum sensitivity and certificate complexity occurs in the hazard-free world.

\paragraph{Applications to learning:} The function learning problem is as follows: We are provided a few input-output pairs and a guarantee that the function is from some family. The goal is to learn the function from as few samples as possible. The sensitivity of a Boolean function plays an important role in this context. It is known that a function with sensitivity $s$ is completely specified by its values on a Hamming ball of radius $2s$~\cite{GNSTW16}. We prove an analogue for hazard-free extensions of Boolean functions.

\begin{restatable}{theorem}{senball}
\label{thm:senball}
A hazard-free extension $f$ that has $\sen_u(f) \le s$ is specified by its values on any Hamming ball of radius $4s$ in $\zuon$.
\end{restatable}

\section{Preliminaries}

For the set $\{0, \un, 1\}$, we say that $0$ and $1$ are stable and $\un$ is unstable. The partial order $\{\un < 0, \un < 1\}$ is the natural partial order associated with this set, called the \emph{instability partial order}. We extend this partial order to strings of fixed length over this set by saying $x\leq y$ if and only if $x_i \leq y_i$ for all $i$.
A string $x \in \zuo^n$ is a resolution of $y \in \zuo^n$ if $y\leq x$. $x$ is a proper resolution of $y$ if in addition $x\neq y$.
A function $f:\zuon \to \zuo$ is \emph{natural} if for all $x, y$ such that $x \leq y$, we have $f(x) \leq f(y)$ \emph{and} $f(x)$ is stable when all bits of $x$ are stable.
Notice that the partial order can be extended to natural functions as $f \leq g$ if and only if $f(x) \leq g(x)$ for all $x$.

We call elements in the set $f^{-1}(b)$ as $b$-inputs for any $b$.

\begin{definition}[Hazard-free Extensions]
    \label{def:hfe}
    For a Boolean function $f$, we define its hazard-free extension $\widetilde{f}$ as: $\forall y \in \zuo^n$:
$$
\widetilde{f}(y) = \begin{cases}
 	b & \textrm{$\exists b \in \zo$ such that 
 	for all resolutions $x \in \zon$ of $y$, $f(y)=b$.}  \\
 	u & \textrm{ otherwise }
 	\end{cases}
$$ 
\end{definition}

When it is clear from the context, we replace $\widetilde{f}$ by $f$ to denote the hazard-free extension of $f$. Notice that the hazard-free extension is the unique maximal natural function that is an extension of the Boolean function. 

A \emph{literal} is either $x_i$ (positive) or $\bnot{x_i}$ (negative) for some variable $x_i$. We consider the two standard Boolean algebra operations called AND, or conjunction, denoted by $\cdot$ or simple juxtaposition; and OR or disjunction, denoted by $+$. Boolean formulas are constructed by composition of these operations in the natural fashion. Given a formula $F$ over these operations and an input $x\in\{0,\un,1\}^n$, we use $F(x)$ to denote the value of $F$ on $x$, where the AND and OR are now replaced with their hazard-free extensions. Ikenmeyer~et~al.~\cite{ikenmeyer2019complexity} showed that the class of functions thus computed is exactly the class of natural functions.

\begin{definition}[Implicants and Implicates]\label{def:pi}
    A \emph{term} is a conjunction of literals and a \emph{clause} is a disjunction of literals. A term $t$ is an \emph{implicant} if $t(x)=1$ implies that $f(x)=1$. A clause is an \emph{implicate} if $c(x)=0$ implies that $f(x)=0$. An implicant or implicate is \emph{prime} if it is minimal in the number of literals. The \emph{size} of an implicant or implicate is the number of literals in it.
\end{definition}

An implicant or an implicate for an $n$-input Boolean function can be represented by a string in $\{0,\un,1\}^n$ such that bit $i$ is $1$ if $x_i$ occurs positively, $0$ if it occurs negatively, and $\un$ otherwise. Given this representation, we also note that the hazard-free extension of $f$ is the unique natural function that is an extension of $f$ that outputs $1$ on all the prime implicants and $0$ on all the prime implicates.

A formula is in Conjunctive Normal Form (CNF) if it is an AND of clauses and is in Disjunctive Normal Form (DNF) if it is an OR of terms. A CNF formula is a $k$-CNF if each clause has size at most $k$ and a DNF formula is a $k$-DNF if each term has size at most $k$. A circuit or a formula for a function $f$ is hazard-free if it computes the hazard-free extension of $f$ assuming the operations AND, OR, and negation in the circuit or formula compute the hazard-free extensions of AND, OR, and negation.

\section{Depth Bounds for Hazard-free Decision Trees}

In this section, we prove various lower and upper bounds for depth of hazard-free decision trees.
\subsection{$\fmux$ is Evasive in the Hazard-free World}

A function is \emph{evasive} if the decision tree depth is the maximum possible. We first prove Theorem~\ref{thm:muxdepth}, i.e., $\fmux_n$ is evasive for all $n$. In fact, we prove a more general theorem. We consider hazard-free decision trees for $\fmux_n$ that are guaranteed to produce the correct output when the number of unstable values in the input is at most $k$, for an arbitrary $k\in[0,2^n+n]$. We call such decision trees, \emph{$k$-bit hazard-free decision trees}.

\begin{theorem}
For $k \in [0, n]$, any optimal depth decision tree that correctly computes the hazard-free extension of $\fmux_n$ on inputs with at most $k$ unstable values has depth $2^k + n$.
\end{theorem}
\begin{proof}
First, we prove the lower bound. Let $T$ be a $k$-bit hazard-free decision tree that computes $\fmux$. We describe an adversarial strategy, i.e., the values returned by the responder so that the querier has to query $2^k+n$ bits to deduce the correct answer. At any step in this process, based on the bits that are queried, we define $Q_s$, (resp. $Q_d$) to be the set of selector bits (resp. the data bits) that have been already been queried. We define $C_{s}$ to be the set of selector bit that were queried and were answered to be $b \in \{0,1\}$, and finally $U_{s}$ as the set of selector bits that have been either not queried so far or were queried and were answered to be $\un$. Similarly, we define $U_{d}$, a subset of data bits as $\{x_{(b_0,\dotsc,b_{n-1})} \mid b_i=r_i \text{ for all $i$ in }  C_s \}$. Note that $|U_{d}| = 2^{|U_s|}$.

The adversarial strategy is as follows: Initialize an index into data bits $r = (1, \dotsc, 1)$. When a data bit $x_i$ is queried and if $U_d \setminus Q_d = \{x_i\}$, then set $r=i$ and return $0$, else just return $1$. When a selector bit $s_i$ is queried, if $|Q_s| < k$, then return $\un$, else return $r_i$. Notice that the responder returns at most $k$ unstable values and all data bits are stable. Also, the adversary returns at most $n-k$ stable values for selector bits. This implies that $|U_d| \geq 2^k$ at all times.

The adversary maintains the following invariants: (1) There is at least one way to assign values to the unqueried selector bits so that any data bit in $U_d$ is indexed by one of the resolutions of the resulting selector bits. (2) If there are any unqueried selector bits, then for any data bit $x_i$ in $U_d$, there exists two settings for the values of the unqueried selector bits so that $x_i$ is indexed by a resolution of one of the settings and $x_i$ is not indexed by any resolution of the other setting.

Part (1) is true by definition of $U_d$. For part (2), observe that $U_d$ contains all data bits for which the indices are consistent with the values of queried selector bits with stable values. Say $s_j$ is an unqueried selector bit and $x_{(b_0,\dotsc,b_{n-1})}\in U_d$. Then, by setting $s_j=\bnot{b_j}$, we can produce a setting that avoids indexing this data bit. On the other hand, by setting all unqueried selector bits $s_j$ to $b_j$, we can ensure that this data bit is indexed by a resolution.

We claim that at the end of the game, all data bits in $U_d$ must have been queried. Suppose there is an unqueried data bit in $U_d$. By part (1) of the invariant, we know that we can set any remaining selector bits so that the unqueried data bit is indexed by some resolution. We know that the querier received a value of $1$ for all the queried data bits. So if they answer $1$, then we can say that one of the unqueried data bits is $0$ and the output cannot be $1$. If they answer $\un$ or $0$, then we can say that all unqueried data bits are $1$ forcing the output to be $1$. So at least $2^k$ data bits have to be queried.

We claim that all the selector bits must be queried. Suppose not. Then, we know that exactly one data bit in $U_d$ at the end has value $0$ and the rest have value $1$. By part (2) of the invariant, we can set the remaining selector bits in two ways so that this $0$ data bit is included, in which case the output is $\un$, or the $0$ data bit is excluded, in which case the output is $1$. So the querier has to query all $n$ selector bits.
\end{proof}

Notice that when $k=n$, we have to query all bits in the input. So for $k\in [n+1, 2^n+n]$, the depth is exactly $2^n+n$.

\subsection{Functions with Low Depth Hazard-free Decision Trees}
The theory of decision tree depth for hazard-free functions would not be very interesting if all functions were hard. We now show that there are some very easy, but non-trivial functions.
\smalldepth*
\begin{proof}
    The function has $n$ input variables $\mathbf{s} = (s_0, \dotsc, s_{n-1})$ and $2^{n+1}-1$ input variables $y = (y_{(b_1, \dotsc, b_i)})$ where $i\in [0, n+1]$ and each $b_j\in\{0,1\}$ (The index set for $y_{*}$ variables is sequences of bits of length $0$ through $n+1$):
    \begin{equation*}
        f_n(\mathbf{s}, \mathbf{y}) = \fmux_n(s_0, \dotsc, s_{n-1}, g_{(0, \dotsc, 0)}(\mathbf{y}), \dotsc, g_{(1, \dotsc, 1)}(\mathbf{y}))
    \end{equation*}
    where we define:
    \begin{equation*}
        g_{(v_0, \dotsc, v_{n-1})} = z_{\mathbf{v}0}\ \mathsf{op}_{\mathbf{v}0}\  (z_{\mathbf{v}1}\ \mathsf{op}_{\mathbf{v}1}\ (z_{\mathbf{v}2} \dotsc \mathsf{op}_{{\mathbf{v},n-1}}\ z_{\mathbf{v}n})\dotsc)
    \end{equation*}
    where $\mathbf{v} = (v_0, \dotsc, v_{n-1})$, $\mathsf{op}_{\mathbf{v}j}$ is Boolean $+$ if $v_j = 0$ and Boolean $\cdot$ otherwise. Each $z_{\mathbf{v}j}$ is the variable $y_{(v_0,\dotsc,v_{j-1})}$. In particular, $z_{\mathbf{v}0}$ is the variable $y_{()}$. i.e., it is independent of $\mathbf{v}$.

    We prove the theorem by induction on $n$. For $n=1$, the function $f_1$ is (after rewriting variables for convenience): $f_1(s, x, y, z) = \fmux_1(s, x+y, xz)$. A hazard-free decision tree of depth $3$ for $f_1$ is as follows: First, we query $s$. If it is $0$ or $1$, then two more queries suffice. Suppose $s = \un$. Then, we query $x$. If $x=0$, then $xz=0$ and the answer can be determined by only querying $y$. If $x=1$, then $x+y=1$ and the answer can be determined by only querying $z$. If $x=\un$, then we can immediately answer $\un$ since $x+y$ can only be $\un$ or $1$ and $xz$ can only be $\un$ or $0$.

\begin{figure}
\caption{Hazard-free decision tree for $f_n$\label{fig:smalldepth}}
\begin{tikzpicture}[grow=right]
    \tikzset{level distance=2cm}
    \Tree [.{$s_0$}
      [.{$y_{()}$} [.{$T_{n-1}$} ] [.{$T''_{n-1}$} ] [.1 ] ]
      [.{$y_{()}$} [.{$T'_{n-1}$} ] [.{$\un$} ] [.{$T''_{n-1}$} ]]
      [.{$y_{()}$} [.0 ] [.{$T'_{n-1}$} ] [.{$T_{n-1}$} ] ]
    ]
\end{tikzpicture}
\end{figure}
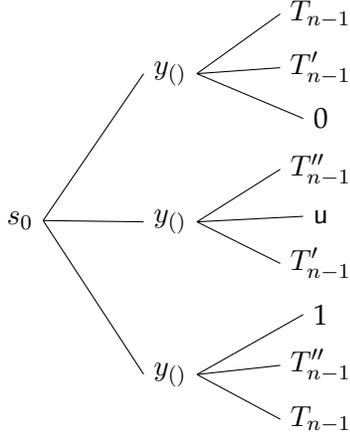

    The proof for the inductive case is similar to the base case. For $f_n$, in two queries, we either determine the answer or reduce the remaining decision to the $n-1$ case. The decision tree is given in Figure~\ref{fig:smalldepth}. We split the proof of correctness of our construction into nine different cases corresponding to the leaves of this tree.
    
    ($s_0=0$ and $y_{()}=0$) In this case, we know that the selector bits can only select functions $g_{\mathbf{v}}$ such that $\mathsf{op}_{\mathbf{v}0} = +$. Since, $y_{()}=0$, the remaining problem is isomorphic to $f_{n-1}$ on the selector bits $(s_1, \dotsc, s_{n-1})$ and the functions $g_{\mathbf{v}}$ where $v_0=0$ and the variable $y_{()}$ has been substituted with $0$.

    ($s_0=0$ and $y_{()}=\un$) We know that the selector bits can only select functions $g_{\mathbf{v}}$ such that $v_0=0$ which implies $\mathsf{op}_{\mathbf{v}0} = +$. Since $y_{()}=\un$, the function $f_n$ can only now evaluate to $\un$ or $1$. More precisely, it evalutes to $1$ if and only if all $g_{\mathbf{v}}$ that are selected evaluate to $1$ and $\un$ otherwise. Let $X_{n-1}$ be the hazard-free decision tree for the $(n-1)$ case where the selector bits are $(s_1,\dotsc,s_{n-1})$ and they select functions $g_{\mathbf{v}}$ where $v_0=0$ and $y_{()}$ and the first $+$ are removed from the functions $g_{\mathbf{v}}$. It is easy to see that this is isomorphic to $f_{n-1}$. The tree $T''_{n-1}$ is obtained by relabeling leaves labeled $0$ in $X_{n-1}$ with $\un$.

    ($s_0=0$ and $y_{()}=1$) The output of $f_n$ must be $1$. So we can output $1$ immediately.

    ($s_0=\un$ and $y_{()}=0$) We know that selector bits must select some $g_{\mathbf{v}}$ where $v_0=0$ and some $g_{\mathbf{w}}$ where $w_0=1$. All such $g_{\mathbf{w}}$ must evaluate to $0$ as $\mathsf{op}_{\mathbf{w}0} = \cdot$. So the value of $f_n$ is $0$ or $\un$ and it is $0$ if and only if all such $g_{\mathbf{v}}$ evaluate to $0$. We can determine this using a decision tree that is isomorphic to that for $f_{n-1}$ where all leaves labeled $1$ are relabeled with $\un$.

    ($s_0 = y_{()} = \un$) In this case, we select both $g_{\mathbf{v}}$ where $\mathsf{op}_{\mathbf{w}0} = +$ and $\mathsf{op}_{\mathbf{w}0} = \cdot$. Since $y_{()} = \un$, the output of $g_{\mathbf{v}}$ where $\mathsf{op}_{\mathbf{w}0} = +$ cannot be $0$ and the output of $g_{\mathbf{v}}$ where $\mathsf{op}_{\mathbf{w}0} = \cdot$ cannot be $1$. Therefore, the output has to be $\un$.

    ($s_0 = \un$ and $y_{()} = 1$) We can deduce that $\mathsf{op}_{\mathbf{w}0}$ for selected $g$ could be $+$ or $\cdot$. Since $y_{()} = 1$. The functions $g$ where this operation is $+$ must evaluate to $1$. Therefore, the remaining decision tree only has to deal with the other half.

     ($s_0 = 1$ and $y_{()} = 0$) Similar to $s_0 = 0$ and $y_{()} = 1$.

     ($s_0 = 1$ and $y_{()} = \un$)  Similar to $s_0 = \un$ and $y_{()} = 0$.

     ($s_0 = 1$ and $y_{()} = 1$)  Similar to $s_0 = 0$ and $y_{()} = 0$.
\end{proof}
To the best of our knowledge, the functions $f_n$ we define in the above proof has not been described anywhere else. It is unknown whether these functions have any other applications. Our interest in this function in this work is purely to show that easy functions exist in the hazard-free decision tree model.

\section{Size Bounds for Hazard-free Decision Trees}

\subsection{Bounds for $\fmux$ function}
We have seen that the depth of any hazard-free decision tree for {\sf MUX} has a depth of $2^n + n$, in contrast to the Boolean decision tree depth of $n+1$. The size of the decision tree in the Boolean case is $2^{n+1}$, which is linear in the number of inputs. In this subsection, we shall bound the size of the hazard-free decision tree for $\fmux_n$ by $\Theta(4^n)$, which is quadratic in the input size.

\muxsize*
\begin{proof}{(of lower bound)}

Consider any hazard-free decision tree $T$ computing $\fmux_n$. We only consider leaves in $T$ such that on the path from the root to the leaf, all selector bits are queried. We partition this set of leaves into $L_\alpha$ for $\alpha\in\zuon$ such that any leaf in $L_\alpha$ is only reached by inputs where the selector bits have value $\alpha$. Since we only consider leaves where all selector bits are queried, this is a partition of such leaves.

\begin{lemma}\label{lem:lalpha}
    If $\alpha$ contains exactly $k$ unstable bits, then $|L_\alpha| \geq 2^{k+1}$.
\end{lemma}

We complete the proof assuming the above lemma. The calculation is shown below:
\begin{align*}
    |L| &\geq \sum_{\alpha \in S}|L_\alpha| \\
    &\geq \sum_{k=0}^n 2^{k+1} |\{\alpha \mid \alpha \text{ has exactly $k$ unstable values.}\}|\ \text{Lemma~\ref{lem:lalpha}}\\
    &= \sum_{k=0}^n2^{k+1}{n \choose k}2^{n-k}\\
    &= 2 \cdot 2^n \sum_{k=0}^n {n \choose k} = 2 \cdot 4^n
\end{align*}

\begin{proof}(of Lemma~\ref{lem:lalpha})
Define $I_\alpha$ to be the set of inputs consistent with $\alpha$ on the selector bits such that exactly one data bit indexed by the resolutions of $\alpha$ is $0$ and all other data bits is $1$. These are all $\un$-inputs. Note that $|I_\alpha| = 2^k$. We claim that for any input $z\in I_\alpha$, the tree $T$ must query all selector bits and the data bit (say $x_{(b_0,\dotsc,b_{n-1})}$) with the value $0$ of $z$. Suppose a selector bit $s_j$ is not queried and $\alpha_j=\un$. Consider $z'$ that is the same as $z$ except $s_j=\bnot{b_j}$. The output on $z'$ for $T$ will also be $\un$ when it should be $1$. Suppose $\alpha_j$ is stable. Consider $z'$ such that it differs from $z$ only by $s_j=\bnot{b_j}$. The output on $z'$ is $1$ but $T$ reaches the same leaf for $z$ and $z'$. This establishes that all inputs in $I_\alpha$ are in $L_\alpha$. We also claimed that $x_{(b_0,\dotsc,b_{n-1})}$ is also queried. Indeed, if it is not, we can construct a $z'$ that is exactly the same as $z$ where that data bit is $1$ that will follow the same path. But that input is a $1$-input.

Notice that two distinct inputs in $I_\alpha$ must reach different leaves in $L_\alpha$. This is because they have unique data bits that are $0$ and they are distinct for distinct inputs. The above claim shows that both these bits have to be queried for those inputs. Therefore, they cannot reach the same leaf. So $|L_\alpha| \geq 2^k$.

If we consider input set $J_\alpha$ which is similar to $I_\alpha$ except that exactly one indexed (by $\alpha$) data bit is $\un$ and the rest are $1$, we can argue the same for $J_\alpha$ and conclude that there must be $2^k$ inputs that reach distinct leaves. Moreover, these must reach distinct leaves from $I_\alpha$ as no data bit in an input from $I_\alpha$ has value $\un$.\end{proof}\end{proof}

We remark that the above argument also gives a lower bound on the size of $k$-bit hazard-free decision trees for $\fmux$. Note that Lemma~\ref{lem:lalpha} is already parameterized in terms of the number of unstable bits. By a similar argument, the following lower bound follows:
any $k$-bit hazard-free decision tree for $\fmux$ must have at least $2^n \left(\sum_{i=0}^k{n \choose i}\right)$.

\begin{proof}{(of upper bound)}
We can show that $\fmux_n$ has a hazard-free decision tree with $4^{n+1}-3^n$ leaves. First, we query all the selector bits in order. This results in a tree with $3^n$ leaves where each label is naturally associated with $\alpha\in\zuon$ that correspond to the values of the selector bits. For each $\alpha$, we will now proceed to query all the data bits that are indexed by resolutions of $\alpha$. If $\alpha$ has exactly $k$ unstable values, then we will query at most $2^k$ data bits. Let us call this tree (a subtree of the main decision tree) $T_\alpha$. We will have three subtrees for $T_\alpha$, which we call $T_{\alpha,0}$, $T_{\alpha,\un}$, $T_{\alpha,1}$ corresponding to the values of the first data bit. Notice that $T_{\alpha,\un}$ is just a leaf labeled $\un$. From $T_{\alpha,1}$, we query the next data bit. If it is $0$ or $\un$, we immediately output $\un$. Otherwise, we build the tree $T_{\alpha,11}$. Notice that the leaves of $T_\alpha$ correspond exactly to strings:
\begin{enumerate}
    \item $z\in\zuo^{\leq 2^k}$ such that $z=x\un$ where $x$ is either all $1$ or all $0$ (possibly empty). When $x$ is empty, these are the same string. So there are $2^{k+1}-1$ such strings.
    \item $z\in\zuo^{\leq 2^k}$ such that $z=x0$ where $x$ is non-empty and all ones or $z=x1$ and $x$ is non-empty and all zeroes. There are $2^{k+1}-2$ such strings.
    \item Two more strings $z=1^{2^k}$ or $z=0^{2^k}$.
\end{enumerate}
The total is $2^{k+2} - 1$. Summing over all $\alpha$ gives us the size upper bound of $4^{n+1} - 3^n$.
\end{proof}


\subsection{Constructing Hazard-free Decision Trees from Boolean Decision Trees}

In this subsection we consider a bottom up construction of a hazard-free decision tree from the original decision tree of a Boolean function, say $T$. In our bottom-up approach, we shall start building the $u$ subtree for every node from that last level.

\dthfdt*
\begin{proof}
(first inequality) Given a hazard-free decision tree $T'$, consider $T$ obtained by removing all subtrees from $T'$ that can only be reached by following an edge labelled $\un$. It is easy to see that $T$ is a decision tree computing $f$, and hence $\size(T) \ge \size(f)$. Since $T$ has at least one more leaf for each internal node of $T'$ (corresponding the query response of $\un$ at that node), we have that $\size(T') \ge 2\size(T)-1 \ge 2\size(f)-1$. This inequality is tight as witnessed by the parity function on $n$ inputs.

(second inequality) We describe how to construct a hazard-free decision tree $T'$ from a Boolean decision tree $T$.\\[-2mm]

\noindent{\bf Construction of $T'$}:
Let $v$ be a node querying a variable $x_v$ and its subtrees when $x_v = 0$ (resp. $x_v =1$) are $T_{v, 0}$ (resp. $T_{v, 1}$) in the original Tree $T$. Let $v_0$ and $v_1$ be the children of node $v$ in the decision tree $T$ through a bottom-up inductive approach where we add $\un$ query subtrees to every node of $T$.

Let us assume inductively that when a $\un$ query subtree is added to any node of $T$, all of the children of the node already have their $\un$ query subtrees attached. The base case would be at the leaves where you do not have to add anything since there's no variable to be queried. 

The inductive step at a node $v$ would be as follows: by induction we assume to have converted the decision tree $T_{v, b}$ rooted at $v_b$ to $T_{v, b}'$ for $b \in \{0, 1\}$. Let $\rho(v)$ be the set of tuples of variables and the values that they are set to in the path from root to the nodes $v$. For $b \in \{0,1\}$, and any leaf $\ell \in T_{v, b}'$, let $\rho_b(\ell)$, the set of tuples of variables and the values that they are set to in the path from $v_b$ to the leaf $\ell$.

We will now construct a subtree $T'_{v, \un}$ of $v$ for the case $x_v = \un$. The subtree $T'_{v, \un}$ is a copy of $T_{v, 0}'$ with the leaves labelled $0$ and $1$ of $T_{v, 0}'$ replaced with copies of $T_{v, 1}'$ appropriately modified as follows: For $b \in \{0,1\}$, let $\ell$ be the leaf in the new copy of $T_{v, 0}'$ labelled with $b$. Replace the leaf $\ell$ with a copy of $T'_{v, 1}$ with the following modifications (1)~ Remove any branches that are inconsistent with assignments in $\rho_0(\ell)$ (2)~ Replace the $\overline{b}$-leaves with $\un$.\\

\noindent \textbf{Correctness:} We now prove the correctness of the above construction as follows:
We show that for any input $x \in \zuon$ that satisfies $\rho(v)$, the subtree constructed above ensures that the leaf that the subtree reaches has the value $f(x)$. The proof is by induction on the height of $v$.

{(Base case)} The base case is when $v$ is a leaf. In this case, the construction is trivially correct.

{(Inductive case)} Consider a non-leaf $v$, For a child $w$ of $v$, we can now assume that any input $x \in \zuon$ that satisfies $\rho(w)$ reaches the value $f(x)$ in the subtree that has been constructed at $w$ (including $T_{w, 0}', T_{w, 1}', T_{w, \un}'$). Any $x$ consistent with $\rho(v)$, and $x_v$ is $0$ or $1$, the tree is lead to the subtrees rooted at $v_0$ or $v_1$, which gives the correct value by induction hypothesis.

Now we consider the case of $x$ consistent with $\rho(v)$ and $x_v = u$. 
Let $x_0$, and $x_1$ be the inputs obtained by replacing the $x_v$ with $0$ and $1$ respectively. We know that:
\[
    f(x) = \begin{cases}
        b & \textrm{ if $f(x_0) = f(x_1) = b$} \\
        u & \textrm{ otherwise} 
    \end{cases}
    \]

We now have two cases based on value of $f(x_0)$.
\begin{description}
    \item{\bf Case 1:} $f(x_0) = b$ where $b \in \{0,1\}$. We know that $x_0$ leads the tree to a leaf $\ell_0$, labelled $b$, via the assignment $\rho(v) \cup \rho_0(\ell) \cup \{(x_v,0)\}$ in the tree $T'_{v, 0}$. The input $x_1$ leads the tree to a leaf $\ell_1$ in $T_{v, 1}'$ and the corresponding path is represented by the assignment $\rho(v) \cup \rho_1(\ell_1) \cup \{(x_v,1)\}$. 
    
    Consider $x$ in $T'$. It reaches the node $v$ in $T$ (since it is consistent with $\rho$) and it further reaches $T_{v, \un}'$ since $x_v=\un$. From there, the input $x$ will follow the path to the node corresponding to the assignment $\rho(v) \cup \{(x_v,u)\} \cup \rho_0(\ell_0)$ since $T_{v, \un}'$'s first part is same as $T_{v,0}'$ by the construction. Further, since the leaf ($\ell_0$) in the copy of $T_{v,0}'$ has been replaced with a modification of $T_{v,1}'$ restricted to $\rho_0$, $x$ follows the path to the leaf indexed by $\rho(v) \cup \{(x_v,u)\} \cup \rho_0(\ell_0) \cup \rho_1(\ell_1)$ and reaches the new leaf in the modified copy of $T_{v, 1}'$. In this case, the output would be 
    \[
    T'(x) = \begin{cases}
        b & \textrm{ if $f(x_1) = b$} \\
        u & \textrm{ otherwise} 
    \end{cases}
    \] as desired.
    \item{\bf Case 2:} $f(x_0) = u$. We know from induction hypothesis that $x_0$ follows the path to a leaf ($\ell$) labelled $u$ via the assignment $\rho(v) \cup \{(x_v,0)\} \cup \rho_0(\ell)$. Therefore, the input $x$ firstly reaches $T_{v,\un}'$ and since this subtree is a copy of $T_{v, 0}'$, the input $x$ will reach a leaf labelled $\un$ by following a similar path as $x_0$.
\end{description}
\noindent \textbf{Size of $T'$}:
An easy upper bound on $\size(T'_{v, \un})$ is $\size(T'_{v, 0})\size(T'_{v, 1})$. we have:
\begin{eqnarray*}
\size(T_v') & \le & \size(T_{v, 0}')+\size(T_{v, 1}')+\size(T_{v, \un}') \\
& \le & \size(T_{v, 0}')+\size(T_{v, 1}')+\size(T'_{v, 0})~\size(T'_{v, 1}) \\
& \le & (\size(T'_{v, 0})+1)(\size(T'_{v, 1})+1)-1
\end{eqnarray*}
Hence for any node, we have the following recursive relation for $\size$ for the tree rooted at the node. $\size(T')+1 \le (\size(T_{r, 0}')+1)(\size(T_{r, 1}')+1)$. Notice that the two terms on the right-hand side has the same form as the expression on the left-hand side. Since the construction is recursive, we can keep recursively expanding all the terms on the right-hand side until we reach the leaves. At any leaf, this expression is just $(1+1)$ and there will be $\size(T)$ many of them. So $\size(T') + 1 \leq 2^{\size(T)}$ as desired.
\end{proof}

We remark that the lower bound proof is also tight as demonstrated below by $\fand_n$ function. However, there are functions in which $T_{v, u}'$ attached gets simplified significantly and leads to a reduction in size of the resulting decision tree. This is demonstrated by the construction for the hazard-free decision tree for $\fmux_n$.  

Applying the general construction above to the optimal Boolean decision tree for $\fand_n$ gives us a tree of size $2^{n+1}-1$. We will now prove that the upper bound in Theorem~\ref{thm:dthfdt} is tight. Recall that the Boolean decision tree for $\fand_n$ has size $n+1$.
\andsize*
\begin{proof}
To prove the lower bound, let $T$ be a hazard-free decision tree with optimal size computing $\fand$. Let $L$ be the set of leaves in $T$. For every leaf $\ell$ of $T$ we will associate a set $V_\ell$ which is the set of pairs $(v,\alpha)$ where $v \in V_\ell$ and $\alpha \in \zuo$ such that the computation path that leads to the leaf $\ell$ sets variable $v$ to $\alpha$ for all $(v,\alpha) \in V_\ell$.

Consider $T$ restricted to only $1, u$ which is a binary tree say $T_{1u}$. We prove that $\size(T_{1u}) \ge 2^n$. We know that there are at least $2^n$ leaves and $2^{n}-1$ internal nodes present in $T_{1u}$. Note that in $T$, every internal node will have a subtree (having at least 1 leaf) where the queried variable is set to 0. Hence there are $2^n -1$ more leaves in $T$ than in $T_{1u}$. Hence we have $\size(T) \ge 2^n + 2^n -1$ and we can use the trivial construction to get a hazard free decision tree of this size.

Similar to the size lower bound proof for $\fmux$, we shall choose a set of inputs (of size $2^n$) which should all definitely reach different leaves. Let $L$ be the set of leaves in $T_{1u}$. For every leaf $\ell$ of $T_{1u}$, consider the set of inputs $I = \{x \mid x \in \{1, u\}^n \}$ which does not have $0$. We claim that if $x \in I$ and $T$ reaches the leaf $\ell$ on input $x$, then $V_{\ell}$ has all variables of $x$. Assume this claim, the argument is complete since we know that every $x \in I$ goes to a different leaf and $|L| \ge 2^n$ thereby proving the proposition.

To prove the claim, suppose $x_i$ is not contained in $V_\ell$. We split the proof into two cases. If $T$ outputs $1$ or $u$, consider an input $x'$, such that $x'_i = 0$ and $x'_j = x_j \forall j \neq i$, note that this reaches $\ell$ as well and $\fand(x') = 0 \neq T_{1u}(x') = T(x')$, a contradiction. If $T$ outputs $0$, then observe that $\fand$ evaluates to $1$ on $x$, a contradiction.
\end{proof}

\subsection{Size bounds for $k$ Hazard-free Decision Trees}

We shall explore the size upper bounds of the $k$ hazard-free decision tree $T^{k}$ by a similar construction from the Boolean decision tree by adding $u$ subtree's at every node of $T$. A $k$ hazard-free decision tree computing $f$ should lead to the value $f(x)$ for any $x \in \zuon$ that has at most $k$ uncertain bits. 

\kbitdthfdt*
\begin{proof}
Assume that we have already constructed trees $T^0, T^1, T^2, ..., T^{k}$ where $T^i$ is the $i-$hazard-free decision tree constructed by adding $T^i_{v, \un}$ subtrees to all the nodes $v$ of T. The base step is just $T^0 = T$.

\noindent{\bf Construction of $T^{k+1}$:}
We construct a subtree $T^{k+1}_{v,\un}$ to link to the node $v$ when the query for variable $x_v$ returns $u$. We will build the subtree $T^{k+1}_{v, \un}$ from a copy of $T^k_{v, 0}$ by replacing the leaves labelled $0$ and $1$, with copies of $T^k_{v, 1}$ appropriately modified as follows: for $b \in \{0,1\}$, let $\ell$ be the leaf in the new copy of $T^k_{v, 0}$, labelled with $b$. Replace the leaf $\ell$ with a copy of $T^k_{v, 1}$ with the following modifications (1)~ remove any branches that are inconsistent with assignments in $\rho_0(\ell)$ (2)~ replace the $\overline{b}$-leaves with $u$.

\noindent \textbf{Correctness:} Assuming correctness of $T^k$ we can prove the correctness of this construction almost exactly as in the hazard-free decision tree construction.

Consider any input $x \in \zuon$ such that the number of $u's$ is $x$ is at most $k+1$. When the $T^{k+1}$ is instantiated with $x$ let $v$ be the first node that is set to $x_v = u$. If such a node is not encountered then the $x$ reaches a node of the original tree and it is easy to prove that all of it's resolutions reach that in $T$ and hence $T^{k+1}$ outputs $f(x)$. Otherwise let $x_0$ and $x_1$ be inputs obtained by replacing $x_v$ in $x$ with 0 and 1 respectively. We know that 
\[
    f(x) = \begin{cases}
        b & \textrm{ if $f(x_0) = f(x_1) = b$} \\
        u & \textrm{ otherwise} 
    \end{cases}
\]

Assuming the correctness of $T^k$ and since the number of $u'$s in both $x_0$ and $x_1$ is atmost $k$, the same case wise proof as in the correctness of the hazard-free decision tree construction goes exactly when $T'_{v, 0}$ and $T'_{v, 1}$ are replaced by $T^k_{v, 0}$ and $T^k_{v, 1}$ respectively.

Note that unlike the previous construction however, this need not be bottom-up as long the above is done for every node since we only need $T^k$ for constructing the $\un$ subtree and none of the subtrees from $T^{k+1}$ itself.

\noindent \textbf{Size Analysis:}

Note that we $L(T^{k+1}_{v, \un})$ can be upper bounded by $L(T^k_{v, 0})L(T^k_{v, 1})$ and we have 

$$L(T^{k+1}) \le L(T) + \sum_{n \in \mathcal{I}(T)}L(T^{k+1}_{n, \un})$$
where $\mathcal{I}(T)$ is the set of internal nodes of $T$ and $T^{k+1}_{n, \un}$ is the $x_N = \un$ subtree added to $T$ in order to construct $T^{k+1}$

\begin{eqnarray*}
L(T_{k+1}) & \le & L(T) + \sum_{n \in \mathcal{I}(T)}L(T^{k}_{n, 0})L(T^{k}_{n, 1}) \\
& \le & L(T) + \sum_{n \in \mathcal{I}(T)}L(T^{k}_{n, 0}) L(T^{k}_{n, 1}) \\
& \le & L(T) + \sum_{n \in \mathcal{I}(T)}{(L(T^k) - 1)}^2 \\
& \le & L(T) + \sum_{n \in \mathcal{I}(T)}({L(T^k)}^2 - 1) \\
& \le & L(T){L(T^k)}^2
\end{eqnarray*}

On recursively expanding $L(T^k)^2$ we have $L(T^{k+1}) \le L(T)^{2^{k+2} - 1}$

\noindent \textbf{Depth Analysis:} Let the leaf with the maximum depth in $T^{k+1}$ be $\ell$ with $\rho(\ell)$ its set of tuples of variables and values that they are set to in the path from the root to the leaf. Suppose $\rho(\ell)$ does not have any $\un$ that its variables are set to we are done, since then $d(T^{k+1}) = d(T)$. Else suppose $v$ is the first node whose variable is set to $\un$ in the path from root to $\ell$. We know that 
\begin{eqnarray*}
d(T_{k+1}) & \le & d(v) + d(T^{k+1}_{v, \un}) \\ & \le & d(v) + d(T^{k}_{v, 0}) + d(T^{k}_{v, 1}) \\
& \le & 2 d(T^{k})
\end{eqnarray*}

Recursively expanding gives $d(T^{k+1}) \le 2^{k+1}L(T^k)$. This completes the proof.
\end{proof}


\subsection{Lower Bound for Size in terms of Prime Implicants}

\begin{theorem}
    \label{thm:cm}
    Let $f$ be any function. Then, $\size_u(f) \geq m+M$ where $m$ is the number of prime implicants of $f$ and $M$ is the number of prime implicates of $f$.
\end{theorem}

\begin{proof}
    Note that a prime implicant $P$ exactly corresponds to the input $x_P$ where all the variables that occur positively in $P$ are set to 1, those that occur in a negated form in $P$ are set to 0 and the remaining are set to $\un$. Note that the resolutions of $x_P$ are exactly the inputs that are covered by the prime implicant. Let $\mathcal{X_P}$ be the inputs corresponding to the prime implicants. Since no two of these inputs are the same we know that $|X_P| = m$.

    The theorem now follows from the proof of this claim.
    \begin{claim}
        Suppose $T$ is a hazard-free decision tree representing the Boolean function $f$. Every $x_P \in \mathcal{X_P}$ reaches a leaf $L$ such that every index with a boolean bit in $x_P$ has been queried.
    \end{claim}

    \begin{proof}
    Suppose for the sake of contradiction, $x_P$ reach the leaf $L$, with $V(L)$ (the set of variables it has queried and the value that each has been set to) and a variable in $P$ is not queried. We know that the values of the indices in $V(L)$ have to match $x_P$ i.e ${x_Q}_i = b_i \forall (i, b_i) \in V(L)$. Consider the input $x_L$ such that ${x_L}_i = b_i \forall (i, b_i) \in V(L)$ and the remaining bits are set to $u$. Note that there exists an index $j$ where ${x_P}_j = b_j$ and ${x_L}_j = u$, and hence we know that on flipping certain $u$'s in $x_L$(at least) to Boolean bits we get $x_P$. Note that $x_L$ corresponds to an implicant $L'$ (exactly the boolean bits in $x_L$ are considered) which covers the implicant $P$ which is a contradiction as $P$ is a prime implicant.
    \end{proof}
\end{proof}

\section{Sensitivity Theorem for Hazard-free Functions}

The celebrated sensitivity theorem \cite{Hua19} in Boolean function complexity (see survey~\cite{HKP11})  proves that decision tree depth is polynomially equivalent to other Boolean function parameters such as sensitivity, block sensitivity, and certificate complexity. In this section, we prove an analogue of this theorem for hazard-free functions.

First, we define these parameters and establish that they are interesting in the hazard-free world as well.
For $B \subseteq [n]$, $x \in \{0,1,u\}^n$, define $x \oplus B = \{ y \in \{0,1,u\}^n \mid \forall i \in B, x_i \ne y_i \}$.

\begin{definition}[\textbf{Hazard-free Sensitivity}] \label{def:hfsens}
Let $f$ be the hazard-free extension of some $n$-input Boolean function. For an $x\in \{0, \un, 1\}^n$, we define the sensitivity of $f$ at $x$ as:
\begin{equation*}
    \sen_\un(f, x) = |\{i \mid \text{there is a $y$ such that $f(y) \neq f(x)$ and $y$ differs from $x$ at, and only at, position $i$}\}|
\end{equation*}
The elements $i$ of the set are called the \emph{sensitive bits of $x$ for $f$}. The \emph{hazard-free sensitivity of $f$}, denoted $\sen_\un(f)$ is the maximum sensitivity of $f$ over all $x$.
\end{definition}

As in the Boolean world, it is easy to see that hazard-free sensitivity trivially \emph{lower bounds} hazard-free decision tree depth. Take an $x$ at which sensitivity is achieved and look at the path followed by $x$ on any hazard-free decision tree. This path must read all the sensitive bits. If not, we can change the value of the bit not read and the output of the decision tree will remain the same. In the Boolean world, for a long time, it was open whether sensitivity (or a polynomial factor of it) \emph{upper bounds} decision tree depth. Nisan \cite{Nis89} introduced block sensitivity, a stronger variant of sensitivity, which was polynomially equivalent to decision tree depth. We now define block sensitivity for hazard-free functions.

\begin{definition}[\textbf{Hazard-free Block Sensitivity}]
Let $f$ be the hazard-free extension of a Boolean function. For $x \in \zuon$, the hazard-free block sensitivity of $f$ at $x$ is defined as
maximum $k$ such that there are disjoint subsets $B_1, B_2, \ldots B_k \subseteq [n]$ and for each $i \in[k]$, there is a $y$ such that $f(y) \neq f(x)$ and $y$ differs from $x$ at exactly the positions in $B_i$. The hazard-free block sensitivity of $f$, $\bsen_\un(f)$ is then defined as the maximum block sensitivity of $f$ taken over all $x$.
\end{definition}

Clearly, we have $\sen_\un(f) \leq \bsen_\un(f)$ because we can consider individual positions to be blocks of size 1. We can also see that $\bsen_\un(f) \leq \cc_\un(f)$ because the prover must exchange at least one bit from each sensitive block. If not, we can change all the value in the omitted block to change the output and therefore the verifier will not be convinced. We now prove that sensitivity, block sensitivity, certificate complexity are all equivalent up to a factor of two.

\sbscc*
\begin{proof}
    To prove the first inequality, let $k$ be the number of literals in a prime implicant. We claim that all those bits are sensitive. Consider a bit that has, say, value $0$ in the prime implicant. Then, if flipping it to a $1$ does not change the value, then we can make that bit $\un$ and obtain an implicant that subsumes this prime implicant, which is impossible. The case when a bit is $1$ is similar. Finally, a similar argument can be made for prime implicates as well.

    To prove the last inequality, observe that if the output of the function is a $0$, then the prover can pick a clause in the $k_2$-CNF that evaluates to $0$ and reveal those values to the verifier. If the output is $1$, then the prover reveals the term in the $k_1$-DNF that evaluates to $1$. If the output is $\un$, there must be a term \emph{and} a clause that evaluate to $\un$. Any prime implicant and prime implicate must have at least one common variable. Therefore, the prover can reveal those $k_1+k_2-1$ values. Why should this convince the verifier? Observe that the output cannot be $0$ as there is a term in the DNF that is $\un$, and it cannot be $1$ as there is a clause in the CNF that is $\un$. So the verifier can conclude that the output must be $\un$.
\end{proof}

\begin{remark}
    Notice that Theorem~\ref{thm:sbscc} shows that all three parameters are constant-factor equivalent to the largest prime implicant/prime implicates. In the Boolean world, for certificate complexity we get a tighter characterization in terms of CNF/DNF width, which is analogous to prime implicant/prime implicant size here. On the other hand, in the Boolean world, there is a quadratic separation between sensitivity and block sensitivity~\cite{Rubinstein}.
\end{remark}
Before proving polynomial equivalence of these parameters to hazard-free decision tree depth, we establish that these parameters are distinct from hazard-free decision tree depth. Consider the function $\fmux_n$. We have already established that its hazard-free decision tree depth is $2^n+n$. We will prove that its sensitivity, block sensitivity, and certificate complexity are all only $2^n$.
\begin{proposition}
    For $\fmux_n$ where $n\geq 2$, we have $\sen_\un = \bsen_\un = \cc_\un = 2^n$.
\end{proposition}
\begin{proof}
    First, we show $\cc_\un(\fmux_n) \leq 2^n$. For any input with a stable output, the prover reveals all the stable bits in the selector bits and all the data bits indexed by the selector bits. The number of bits revealed is $\max_{k\in[0,n]}(n-k+2^k)$. If the output is unstable, the prover reveals all selector bits and two data bits with distinct stable values indexed by those selector bits or one data bit that is unstable. The number of bits revealed in this case is at most $n+2$.

    We now argue that $\sen_\un(\fmux_n) \geq 2^n$. Consider the input where all selector bits are $\un$ and all data bits are $1$. All $2^n$ data bits are sensitive for this input as changing any of them to $0$ will change the output from $1$ to $\un$.
\end{proof}

We now prove that decision tree depth is polynomially bounded by certificate complexity thereby proving the sensitivity theorem for hazard-free functions. To make this proof work, we need the following definition:
\begin{definition}
    A function $f : \zuo^n \mapsto  \zuo$ is \emph{weakly hazard-free} if:
    \begin{enumerate}
        \item $f$ is monotone wrt the instability partial order.
        \item For all $x\in\zuo^n$ such that $f(y)$ is the same value $b$ for all $y > x$, we have $f(x) = f(y)$ as well.
    \end{enumerate}
\end{definition}
All hazard-free functions are weakly hazard-free. However, a weakly hazard-free function is not necessarily hazard-free or even natural. For example, the one input function $f(0) = 0$, $f(1) = \un$, and $f(\un) = \un$ is not natural but weakly hazard-free. There are also natural functions that are not weakly hazard-free. For example, $f(0) = 0$, $f(1) = 0$, and $f(\un) = \un$. We can now consider naturally extended definitions of ternary decision trees and certificate complexity for weakly hazard-free functions. In the following proof, we work with these definitions.

\dtcc*
\begin{proof}
We prove the second inequality. We make use of the following variations of $\hfcc$: For $b\in\zuo$, we define $\hfcc^{(b)}$ as the maximum certificate size required for proving an output of $b$. In particular, we define $\hfcc^{(\un)}$ formally for weakly hazard-free functions as follows:

$$
\hfcc^{(\un)}({f}) = 
\begin{cases}
0 & \textrm{ if $\nexists x, {f}(x) = u$ } \\
\max_{x | {f}(x) = u} \hfcc({f}, x) & \textrm{otherwise } \\
\end{cases}
$$

We prove the following claim that implies the theorem.
\begin{claim}
\label{claim:hfcc_vs_hfdt}
For any weakly hazard-free function $f$, we have:
\begin{equation*}
    \hfdt(f) \le \hfcc^{(u)}(f) (\hfcc^{(0)}(f) + 1) (\hfcc^{(1)}(f) + 1)
\end{equation*}
\end{claim}
\begin{proof}
The proof is by induction on the number of variables.
For $f$ on a single variable we split into two cases. When $f(x)$ is a constant, $\hfdt(f) = 0$ and the inequality is trivially true. If $f(x)$ is not a constant. Here $\hfdt(f) = 1$ and $\hfcc^{(\un)}(f) = 1$ since there exists at least one pair of inputs $(x, y)$ such that $f(x) = \un$ and $f(y) \ne \un$.

For the induction step, let $f$ be a weakly hazard-free function on $k+1$ variables. If $f$ is a constant, again the inequality is true similar to base case.
Otherwise, we have $f(\mathbf{u}) = u$ and a minimum certificate for this input which we call $Y$. Let $k=|Y|$.

We will now construct a ternary decision tree that computes $f$. We first construct a full ternary decision tree on the variables in $Y$ and replace the leaf $L_a$ for $a \in \{0, 1, u\}^k$ by a ternary decision tree $T_a$ representing $f_a = f(a_1, a_2, ..., a_k, x_{k+1}, .., x_n)$ given by the inductive hypothesis (We assume that $Y$ contains the first $k$ variables wlog). For this, we must prove that $f_a$ is weakly hazard-free. The following claim proves this and other properties that will finally give us the required bound.

\begin{claim}\label{claim:fa}
    For all $a\in \zuo^k$, the function $f_a$ satisfies:
    \begin{enumerate}
        \item It is weakly hazard-free.
        \item $\cc_\un^{(\un)}(f_a) \leq \cc_\un^{(\un)}(f)$.
        \item $\cc_\un^{(0)}(f_a)$ is either 0 or at most $\cc_\un^{(0)}(f)-1$.
        \item $\cc_\un^{(1)}(f_a)$ is either $0$ or at most $\cc_\un^{(1)}(f)-1$.
    \end{enumerate}
\end{claim}
\begin{proof}
    \begin{enumerate}
        \item Since $f$ is monotone and $f_a$ is just a restriction, it is also monotone. Suppose $x$ is such that there is a $b$ and $f_a(y) = b$ for all proper resolutions $y$  of $x$. Then, we have $f(a, y) = b$ for all $y$. Since    $f$ is weakly hazard-free, this implies $f(a, x) = f_a(x) = b$. This proves that $f_a$ is weakly hazard-free.

        \item We argue that for each $x$ that is a $\un$-input, $\cc_\un^{(\un)}(f_a, x) \leq \cc_\un^{(\un)}(f, ax)$. Let $C$ be a smallest certificate for $x$ for $f_a$. Let $C'$ be a smallest certificate for $ax$ for $f$. Notice that $C'$ after removing the variables in $x_1,\dotsc,x_k$ is a certificate for $x$ for $f_a$. So $|C| \leq |C'|$.

        \item If $f_a$ has no $0$-input, then $\cc_\un^{(0)}(f_a) = 0$ by definition. So consider a $0$-input $x$ for $f_a$. Thus, $ax$ is a $0$-input for $f$. Consider a smallest certificate $C$ for $ax$ for $f$. We argue that $C$ must include at least one of the values $x_1, \dotsc, x_k$. Suppose not. Then, since the prover has not revealed any of those bits, it could be that all of those bits are $\un$. But in this case the output must be a $\un$, not a $0$, as $Y$, which is just $x_1 = \dotsm x_k = \un$ is a $\un$-certificate for $f$. As before, $C$ with variables in $x_1,\dotsc,x_k$ removed is a $0$-certificate for $x$ for $f_a$. So the $0$-certificates for $f_a$ are at least one less than these smallest $0$-certificates for $f$ proving this part of the claim.

        \item This is similar to the previous case.
    \end{enumerate}
\end{proof}

Excluding constant $f$, we have the following remaining cases for $f$:

\begin{description}
\item{\textbf{Case 1} $\hfcc^{(0)}(f) = 0, \hfcc^{(1)}(f) \ge 1$: } We know by Claim ~\ref{claim:fa} that $f_a$ is weakly hazard-free and that $\hfcc^{(0)}(f_a) = 0$ and $\hfcc^{(1)}(f_a) \le \hfcc^{(1)}(f)-1$. We can now replace $L_a$ by a tree $T_a$ obtained from the inductive hypothesis. So our final tree has depth at most $|Y| + \max_a \depth(T_a)$. Applying the inequality in the inductive hypothesis, we get:
\begin{eqnarray*}
\hfdt(f) & \le & 
\hfcc^{(u)}(f) + \max_a \hfcc^{(u)}(f_a) (\hfcc^{(0)}(f_a) + 1) (\hfcc^{(1)}(f_a) + 1) \\
& \le & \hfcc^{(u)}(f) + \hfcc^{(u)}(f) (1) (\hfcc^{(1)}(f)) \\ 
& \le & \hfcc^{(u)}(f)(1 + \hfcc^{(1)}(f))(1)  \\
& \le & \hfcc^{(u)}(f) (\hfcc^{(0)}(f) + 1) (\hfcc^{(1)}(f) + 1)
\end{eqnarray*}
\item{\textbf{Case 2} $\hfcc^{(1)}(f) = 0, \hfcc^{(0)}(f) \ge 1$:} This case is similar to the previous one.
\item{\textbf{Case 3}  $\hfcc^{(0)}(f) \ge 1, \hfcc^{(1)}(f) \ge 1$:} We know by Claim ~\ref{claim:fa} that $f_a$ is a natural function and that 
$\hfcc^{(0)}(f_a) \le  \hfcc^{(0)}(f)-1$ and $\hfcc^{(1)}(f_a) \le \hfcc^{(1)}(f)-1$. Replacing $L_a$ with $T_a$ gives us a tree of depth at most $|Y| + \max_a \depth(T_a)$. By applying the inequality in the inductive hypothesis, we get:
\begin{eqnarray*}
\hfdt(f) & \le & \hfcc^{(u)}(f) + \max_a \hfcc^{(u)}(f_a) (\hfcc^{(0)}(f_a) + 1) (\hfcc^{(1)}(f_a) + 1) \\
& \le & \hfcc^{(u)}(f) + \hfcc^{(u)}(f) (\hfcc^{(0)}(f)) (\hfcc^{(1)}(f))\\
& \le & \hfcc^{(u)}(f)(1 + \hfcc^{(1)}(f) \hfcc^{(0)}(f)) \\
& \le & \hfcc^{(u)}(f) (\hfcc^{(0)}(f) + 1) (\hfcc^{(1)}(f) + 1)
\end{eqnarray*}
\end{description}
\end{proof}
\end{proof}

\section{Improved Bounds for Sensitivity}

In this section, we explore sensitivity and prove the following characterization for sensitivity of hazard-free functions. We use $\sen_\un^{(b)}$ where $b\in\zuo$ to denote sensitivity for inputs that yield a $b$ output.
\begin{theorem}\label{thm:senstr}
    Let $f$ be the hazard-free extension of some Boolean function. $\sen_\un^{(1)}(f, x)$ is maximized when $x$ is a prime implicant, $\sen_\un^{(0)}(f, x)$ is maximized when $x$ is a prime implicate, and $\sen_\un^{(\un)}(f, x)$ is maximized when $x$ contains exactly one $\un$.
\end{theorem}
\begin{proof}
    Consider $x$ that is an implicant but not a prime implicant. Then, there is a prime implicant $y < x$. We have already seen that the stable bits in $y$ are sensitive. So we just have to argue that if bit $i$ is stable in $x$ and unstable in $y$, then $i$ is not sensitive for $x$. Suppose $x_i=0$ and $y_i=\un$. Changing $x_i=\un$ cannot change the output as $y$ has the same output as $x$. Changing $x_i=1$ cannot change the output as then $y$ would not be a prime implicant.

    A similar argument as above holds for prime implicates.

    For the last part, consider an $x$ that has more than one $\un$ such that $f(x)=\un$. Then, there is a $y > x$ that has exactly one $\un$ such that $f(y)=\un$ ($y$ is an edge in the hypercube). Let the strings $w$ and $v$ in $\zon$ be the two proper resolutions of $y$ and $j$ be the unique position where $w$ and $v$ differ. We argue that the sensitive bits for $x$ are also sensitive bits for $y$. Bit $j$ is sensitive for $y$. So consider an $i\neq j$ such that bit $i$ is sensitive for $x$. Note that $y_i$ is a stable value. If $x_i=\un$, then it must be that setting $x_i=y_i$ or $x_i=\bnot{y_i}$ changes the output to a stable value. Setting $x_i=y_i$ cannot yield a stable output as the resulting subcube will still contain the edge $y$. So setting $x_i=\bnot{y_i}$ yields a stable output. This implies that changing bit $i$ of $y$ to $\bnot{y_i}$ will yield a stable output, proving that bit $i$ is stable for $y$ as well. Now, we consider the cases where $x_i$ is stable. It must be that $x_i=y_i$ ($x_i$ cannot be $\bnot{y_i}$ as $y>x$) and say $i$ is sensitive for $x$. Since $f$ is natural, it must be that setting $x_i$ to $\bnot{y_i}$ yields a stable output. But then, changing bit $i$ in $y$ to $\bnot{y_i}$ will also yield a stable output proving that bit $i$ is sensitive for $y$ as well.
\end{proof}

The following theorem bounds the hazard-free sensitivity for $\un$-inputs using Boolean sensitivity, denoted $\sen(f)$ for $f$.

\begin{theorem}
    For any Boolean function $f$, we have $\sen_\un^{(\un)}(\widetilde{f}) \leq 2\sen(f)-1$.
\end{theorem}
\begin{proof}
    From Theorem~\ref{thm:senstr}, we know that $\sen_\un^{(\un)}$ is maximized at an input $x\in\zuon$ that has exactly one $\un$. Let $v$ and $w$ be the two proper resolutions of $x$ such that $f(v) = 0$ and $f(w) = 1$ and they differ only in bit $i$. We will prove that all sensitive bits of $x$ for $\widetilde{f}$ are sensitive for either $v$ or $w$ for $f$. Clearly, bit $i$ is sensitive for $v$ and $w$ for $f$. Consider a $j\neq i$ that is sensitive for $x$ for $\widetilde{f}$. Notice $x_j$ must be stable and $x_j=v_j=w_j$. Say, $x_j=0$. Since the output on $x$ of $\widetilde{f}$ is already $\un$, the reason bit $j$ must be sensitive is because setting $x_j=1$ causes $\widetilde{f}$ to output a stable value, say $0$. This implies that if we set $w_j=1$ in $w$, then the output of $f$ changes from $1$ to $0$. So bit $j$ must be stable for $w$ for $f$. The rest of the cases are similar. The sensitive bits of $x$ for $\widetilde{f}$ is contained in the union of sensitive bits of $v$ and $w$ for $f$. The union has at least $i$ as the common bit. The theorem follows.
\end{proof}

\section{Sensitivity and Hazard-free Functions}

As an analogue of the analysis of structure of low-sensitivity functions in \cite{GNSTW16}, we establish similar properties for hazard free extensions.

\cite{GNSTW16} also constructs circuits of size $n^{O(s)}$ for sensitivity-$s$ functions. This was proved before sensitivity theorem was proved in the Boolean world. Note that since sensitivity conjecture is true in the hazard-free setting, functions having hazard-free sensitivity $s$ already has hazard-free decision trees of depth $s^3$ and this implies hazard-free circuits of size $2^{s^3}$ that computes the function $f$.

Nevertheless, these structural property derived in \cite{GNSTW16} could be of independent interest in the context of hazard-free setting as well.
In this section,
the Hamming distance between two inputs $x, y$ in the 3 valued logic is defined as the number of indices $i$ on which the inputs differ, in other words, $x_i \ne y_i$, $D(x, y) = |\{i \mid x_i \ne y_i \}|$. A Hamming sphere of radius $r$ centered at $x$ is defined as $S(x, r) = \{y \mid D(x, y) = r\}$. A Hamming ball of radius $r$ centered at $x$ is defined as $B(x, r) = \{y \mid D(x, y) \leq r\}$.

\senball*

\begin{proof}
For $f : \zuon \rightarrow \zuo$ which is a hazard free extension of a Boolean function we define {\em neighborhood} of $x$ as 
$N(x) = \{y \mid \exists i~ x_{i} \neq y_{i}, \textrm{ and } \forall j \in [n] \setminus \{i\}, x_{j} = y_{j} \}$. Similarly we define {\em i-neighborhood} $N_{i}(x) = \{y \mid  x_{i} \neq y_{i}, x_{j} = y_{j} \forall j \in \{1, 2, ..., n\} \setminus \{i\} \}$. Note that for any $x \in \zuon$ and $i \in [n]$, we have that $|N(x)| = 2n$ and $|N_{i}(x)| = 2$. 
We first prove the following Lemma.
\begin{lemma} 
\label{thm:neighbors}
For any function $f$ such that $\sen_u(f) = s$, if $S \subseteq N(x)$ where $|S| \ge 4s+1$ then $f(x) =\fplu_{y \in S}(f(y))$ where $\fplu$ outputs the most frequently occurring element in $\zuo$ in the input to it.
\end{lemma}

\begin{proof}
We shall prove this lemma by showing that a set of neighbors with certain properties specifies $f$. Consider $S \subseteq [n]$ such that $|S| = k \ge 2s+1$. Define
$N_{S}(x) = \{y_{i_1}, y_{i_2}, ..., y_{i_k}\}$ such that $\forall i_j \in S$, $y_{i_j} \in N_{i_j}(x)$, we claim that $f(x) = \fplu_{y \in N_{S}}(f(y))$.

To see this claim, note that, in $N_{S}(x)$, $x$ has neighbors in at least $2s+1$ different indices. Let the frequency of occurrences of $0$, $1$ and $u$ over these neighbors be $f_{v_1}, f_{v_2}$, and $f_{v_3}$. Let $f(x) = v_{i}$, we know that it will have the remaining $f_{v_{j}}+f_{v_{k}}$ to be the number of sensitive indices, which can at most be $s$. 
Since $x$ has at least $2s+1$ neighbors in $N_S(x)$, $v_{i}$ must occur at least $s+1$ times among the neighbors in $N_S(x)$. Hence the claim follows.

We will now use the claim to prove Lemma~\ref{thm:neighbors}.
Suppose we have $S \subseteq N(x)$ where $|S| \ge 4s+1$, then by PHP, $S$ contains elements from at least $2s+1$ different $N_{i}$'s. If not, as we know $|N_{i}(x)| = 2$, $N_{i}(x) \cap N_{j}(x) = \phi$ and $N(x) = N_{1}(x) \cup N_{2}(x) .... \cup N_{n}(x)$. If $S \subseteq N(x)$ and $S$ does contains elements from at most $2s$ different $N_{i}$'s then $|S| \le 4s$.

    Hence from the $2s+1$ different $N_{i}$'s we get the index set $S' \subseteq [n]$ and neighbors in each of these index sets $N_{S'}(x) = \{y_{i_1}, y_{i_2}, \ldots, y_{i_k}\}$ such that $|S'| \ge 2s+1$ and $y_{s_i} \in N_{s_{i}}(x) \forall s_{i} \in S'$. We know that $f(x) = \fplu_{y \in N_{S'}}(f(y))$. Suppose $f(x) = v_{i}$, we will now prove that $v_{i} = \fplu_{y \in S}(f(y))$. Let the other values be $v_{j}, v_{k}$. We claim that $f_{v_j} + f_{v_k} \leq 2s$. Suppose not. Then the number of indices that are sensitive are at least $s+1$ by pigeon hole principle, which is a contradiction. Hence $f_{v_j} + f_{v_k} \leq 2s$, $v_{i}$ occurs at least $2s+1$ times in $S$ and hence $v_{i} = \fplu_{y \in S}(f(y))$. This completes the proof of the \cref{thm:neighbors}.
\end{proof}



Now we complete the proof of the theorem. Suppose the values of a hazard-free extension $f$ with $\sen_u(f) \le s$ is known over a Hamming ball $B(x,4s)$ centered at input $x$ and of radius $4s$ in $\zuon$. We will prove through induction that if the values in any ball $B(x,r)$ where $r \ge 4s$ is known, the values at $S(x,r+1)$ is fixed.

For the induction step, suppose $B(x, r)$ where $r \ge 4s$ is fixed and $y \in S(x, r+1)$. We claim that $N(y) \cap S(x, r) = r+1$. Let $\mathsf{Diff}(x, y) = \{i \mid x_i \neq y_i\}$. Note that for every $i \in \mathsf{Diff}(x, y)$ such that $x_{i} \neq y_{i}$, we have a corresponding $z^i$ such that $z^{i}_{j} = y_{j} \forall j \neq i$ and $z^{i}_{i} = x_{i}$. $z^i \in S(x, r)$ and $z^i \in N(y)$. It is easy to see that none of the other neighbors of $y$ lie in $S(x,r)$. 
Hence from Theorem ~\ref{thm:neighbors} we can fix $f(y)$ since the value of $f$ at at least $4s+1$ neighbors are known. Similarly, we can fix all $f(y)$ for every $y \in S(x, r+1)$. This completes the proof.
\end{proof}

\begin{remark}
    Note that in the above arguments, the plurality function $\fplu : \{0, 1, u\}^n \rightarrow \zuo$ can be replaced by the hazard free extension of the $\fmaj$. In all the contexts that the function $\fplu$ has been used, the value ($v \in \zuo$) occurring most frequently in the input occurs strictly more than half the times. Suppose that for an input $x \in \zuon$, if $v = b \in \{0,1\}$ such that it occurs more than $\frac{n}{2}$ times in $x$, we have $\widetilde{\fmaj}(x) = \fplu(x) = b$. When $v=u$, this means that strictly more than $\frac{n}{2}$ bits in the input are $u$ and the two resolutions of $x$ with all the $u$'s replaced by 0's and 1's respectively have 0 and 1 as their outputs. Hence, $\widetilde{\fmaj}(x) = \fplu(x) = u$. Hence, the two functions are equivalent in the above contexts.
\end{remark}

\section{Discussion and Open Problems}

In this work, we looked at the natural extension of decision tree complexity in a setting where input bits could be unknown. We saw that in this setting the depth and size of decision trees could be exponentially larger compared to the Boolean decision tree where all bits have known, stable values. Is it possible to contain this exponential blowup as in fixed-parameter algorithms? i.e., are there natural parameters $\phi$ for Boolean functions such that $\size_\un(f) \leq g(\phi(f)) \size(f)^{O(1)}$ where $g$ is an arbitrary function?

We also showed that an analogue of the celebrated sensitivity theorem holds in the presence of uncertainty too. This result shows that decision tree complexity, a computational measure, is closely related to many structural parameters of Boolean functions. The proof of sensitivity theorem in this setting is considerably different from the proof of sensitivity theorem in the Boolean world. We can parameterize sensitivity and hazard-free decision tree complexity by restricting our attention to inputs that have at most $k$ unstable bits. The setting $k = 0$ gives us the Boolean sensitivity theorem and the setting $k = n$, the hazard-free sensitivity theorem. Can we unify these two proofs using this parameterization? i.e., is there a single proof for the sensitivity theorem that works for all $k$?

\section{Acknowledgements}
We thank anonymous reviewers for their many helpful suggestions that improved this paper.

\bibliographystyle{alpha}
\bibliography{refs}

\appendix

\section{Other Notions of Sensitivity for Hazard-free Extensions}

In this section we explore two other notions of sensitivity for hazard free extensions and also end up proving that these alternative notions are polynomially related to our choice of sensitivity in the 3 valued logic. We know that for a Boolean function, every input $x$ has one neighbor in a dimension and that edge is 'sensitive' when there is a change in value along that edge. In the 3 valued logic, that every dimension now has 2 edges leaves us with a few choices for the notion of Sensitivity for hazard free functions, in Definition ~\ref{def:hfsens} any of the edges being sensitive makes the dimension sensitive.

As an alternative, for an input with a stable output let us consider only stale bits that can be made unstable/ uncertain leading to uncertainty in the output. Note that this means that an edge can only be sensitive to at most one of the vertices.
\begin{definition}[{\bf Stable Sensitivity}]
Let $f$ be the hazard-free extension of some $n$-input Boolean function. For an $x\in \{0, \un, 1\}^n$ such that $f(x)$ is stable, we define the stable sensitivity $\stabs(f,x)$ of $f$ at $x$ as the number of bits that can be made unstable (individually) so that $f(x)$ is unstable. 
$$\stabs(f) = \max_{x : \textrm{$f(x)$ is stable}} \stabs(f,x)$$
\end{definition}


Another natural notion is sensitivity of `stability' denoting the changes from stable to unstable bits (and vice versa) that lead to an output change from a stable bit to an unstable bit (and vice versa).
 For $B \subseteq [n]$, $x \in \{0,1,u\}^n$, we define $x \oplus B = \{ y \in \{0,1,\un\}^n \mid \forall i \in B, x_i = $u$ \iff y_i \neq u , x_i \neq $u$ \iff y_i = u\}$.
\begin{definition}[{\bf Stability Sensitivity}]
Let $f$ be the hazard-free extension of some $n$-input Boolean function. For an $x\in \{0, \un, 1\}^n$, the Stability Sensitivity at the input $x$, denoted by $\slys(f,x)$, is defined as

$$\slys(f,x) = \left| \left\{ i : \begin{array}{l} \exists y \in x \oplus \{i\} \\
 f(x) = u \iff f(y) \neq u , f(x) \neq u \iff f(y) = u
\end{array}
\right\} \right|$$
The Stability Sensitivity of $f$, $\slys(f)$ is then defined as :
$$\slys(f) = \max_{x \in \{0,1,u\}^n} \slys(f,x)$$
\end{definition}



We shall prove that these measures of sensitivity $\stabs(f), \slys(f)$ are at polynomially related to $\sen(f)$ for any hazard free extension $f$. In fact, similar to the relation between $\sen(f), \hfbs(f)$ and $\hfcc(f)$ they are also away from $\sen(f)$ by a factor of at most 2. We use the notation $\stabs_\un^{(b)}, \slys_\un^{(b)}$ where $b\in \{0, 1\}$ to denote sensitivities for inputs that yield a $b$ output and $\slys_\un^{(u)}$ to denote Stability Sensitivity for inputs that yield a $\un$ output.

\begin{claim}
    Let $f$ be the hazard -free extension of some $n$-bit Boolean function. Then, $ \stabs^{(b)}(f) = \slys^{(b)}(f) = \hfs^{(b)}(f)$ for $b \in \{0, 1\}$ and $\slys^{(u)}(f) \le \hfs^{(u)}(f)$
\end{claim}
\begin{proof}
    It is easy to see that for any $x \in \zuon, b \in \{0, 1\}$ such that $f(x) = b, \stabs(f, x) = \slys(f, x) \le \hfs(f, x)$ and for any input $x$ such that $f(x)  = \un, \slys(f, x) \le \hfs(f, x)$.
    From ~\ref{thm:senstr} we know that $\hfs^{(1)}(f)$ occurs at the input $x$ that corresponds to the largest sized prime implicant (or implicate for $\hfs^{(0)}$) or the minimum sized maximal monochromatic subcube where all the stable bits are sensitive. Hence at this input $\stabs^{(1)}(f, x) = \hfs^{(1)}(f, x) = \hfs^{(1)}(f)$ and by a similar argument for $\stabs^{(0)}$ we have $ \stabs^{(b)}(f) = \slys^{(b)}(f) = \hfs^{(b)}(f)$ for $b \in \{0, 1\}$
\end{proof}
Since we already know that $\hfs^{(u)}(f) \le \hfs^{(0)}(f) + \hfs^{(1)}(f) - 1$, we know that $\frac{\hfs(f)}{2} \le \stabs(f) \le \slys(f) \le \hfs(f)$ thus proving that these notions of sensitivity are at most a factor of 2 away from our choice of sensitivity and thus they would all satisfy the Sensitivity Conjecture.

\section{Combinatorial Proof of the Sensitivity Theorem} \label{sec:CombProof}

In this section we will give an alternate proof for theorem ~\ref{thm:sbscc} by using combinatorial arguments on the sub-structures of the hypercube. This involves associating to every input $x \in \zuon$, a natural subcube $C(x)$ that contains all the resolutions of $x$. 
This argument throws different insights compared to the proof that we presented in the main text.

Recall that implicant (or implicate) inputs correspond to the monochromatic 1 (or 0) subcubes and the prime implicants (or implicates) correspond to those monochromatic 1 (or 0) cubes are maximal since they cannot be increased in any dimension without violating the property of monochromaticity. 

\subsection{Sensitivity and Colorful Subcubes}
\begin{theorem}\label{thm:sens}
    Let $f$ be the hazard-free extension of some Boolean function and $k_1, k_2$ are the sizes of the largest prime implicant and implicate respectively. Then 
    \begin{enumerate}
        \item $\sen_\un^{(1)}(f) = k_1, \sen_\un^{(0)}(f) = k_2$
        \item $\sen_\un^{(\un)}(f) \le k_1 + k_2 - 1$
    \end{enumerate}
    
\end{theorem}

We firstly prove the characterization in the first part of the theorem as follows:
\begin{lemma}\label{lem:MonoCube}
    Let $f$ be the hazard-free extension of some Boolean function $\sen_\un^{(1)}(f, x)$ and $\sen_\un^{(0)}(f, x)$ are maximized when $x$ is a largest prime implicant (or implicate respectively),  and $\sen_\un^{(1)}(f) = k_1, \sen_\un^{(0)}(f) = k_2$
\end{lemma} 

\begin{proof}
    Firstly we shall show that the maximum value of $\sen_\un^{(1)}(f, x)$ occurs at a prime implicant. Consider an $x$ such that $C(x)$ is a monochromatic 1 cube which is not maximal (i.e $x$ corresponds to an implicant but not a prime implicant). Suppose of the sake of contradiction that, maximum value of $\sen_\un^{(1)}(f)$ occurs at a $x$, we will now show that maximal 1 monochromatic cube that $C(x)$ belongs to which is represented by $C(y)$ (where $y$ corresponds to a prime implicant) has all the sensitive indices of $x$ to be sensitive. Note that none of the unstable bits in $x$ are sensitive. Observe that:

    \begin{description}
        \item [1.] Suppose $x_i$ is stable and $y_i$ is unstable then $x_i$ is not sensitive. One replacing $x_i$ by anything else to get $x', C(x')$ is still a subcube of $C(y)$. Hence $f(x') = 1$
        \item [2.] If $x_i = y_i = b'$ and $x_i$ is sensitive then $y_i$ would also be sensitive since the same replacement (on $x$ gives $x'$) on $y$ gives a $y'$ such that $C(x')$ is a subcube of $C(y')$\\Therefore $\sen_\un^{(1)}(f, x) \le \sen_\un^{(1)}(f, y)$
    \end{description}

Now we shall show that at input $x$ representing a maximal monochromatic 1 cube ($x$ corresponds to a prime implicant) exactly the stable bits are sensitive. Note that if any unstable bit is made stable, the resulting subcube is a subcube of $C(x)$ and hence the output is still 1. If any of the stable bits are made $u$ we get a non monochromatic subcube as $C(x)$ is a maximal monochromatic subcube and hence the output is $\un$. This implies that the $\sen_\un^{(1)}$ occurs at the maximum sized prime implicant and it's value is the size of that implicant. A similar argument holds for $\sen_\un^{(0)}$.
\end{proof}

We now upper bound $\sen_\un^{(\un)}$ as in the second part of ~\ref{thm:sens} using the notion of a minimal colorful subcube of a Boolean hypercube of $f$, which is  defined as a subcube that has at least one 0 and 1 output over its inputs and has at least one of its two sub cubes is monochromatic in every dimension. It is, in a sense, a dual of the maximal monochromatic subcube which was used to characterize $\sen_\un^{(b)}$. 

If it is represented by $C(x)$ where $x$ is the input whose resolutions exactly correspond to the sub-cube, then note that $f(x) = \un$ and every $\un$ in $x$ is sensitive (can be replaced with $b$ to get $x'$ that represents one of the monochromatic sub cubes in that dimension). We prove the following properties of a minimal colorful cube:

\begin{lemma}\label{lem:ColorfulCube}
    Let $f$ be the hazard-free extension of some Boolean function, then

    \begin{enumerate}
        \item Any minimum colorful cube $C(x)$ has exactly one input whose output is of value $b$ and the rest have an output value of $\Bar{b}$
        \item For any $x$ corresponding to a minimum colorful cube $\hfs(f, x) \le k_1 + k_2 -1$
        \item $\hfs^{(\un)}(f)$ occurs at an input $x$ corresponding to a a minimum colorful cube
    \end{enumerate}
\end{lemma} 

\begin{proof}

    \begin{enumerate}
        \item Since we know that every $u$ in $x$ is sensitive, suppose $C(x_1), C(x_2), .., C(x_{\Delta})$ be the different monochromatic cubes in the dimensions $1, 2, .., \Delta$ by replacing $\un$ in $x$ with $b_1, b_2, ...$(wlog assume that the first $\Delta$ dimensions are $\un$'s for $x$). Observe that
            \begin{description}
            \item [a.] For any $ (i, j), C(x_i) \cap C(x_j) \neq \Phi$. Note that an input obtained by replacing $j^{th} \un$ with $b_j$ and $i^{th} \un$ by $b_i$ in $x$ gives $x'' \in C(x_i) \cap C(x_j)$. This implies that all the $C(x_i)$ have the same color (or the same value $b$)
            \item [b.] Consider resolutions of $x \notin C(x_1) \cup C(x_2) \cup .., C(x_{\Delta})$. For $y = y_1....y_{\Delta}x_{\Delta+1}...x_n$ to not be in $C_{x_i}$, $y_i$ has to be set to $\Bar{b_i}$ and hence there can only be 1 such resolution possible. Since $C(x)$ is colorful, this resolution is the only one which has a different color.
            \end{description}
        \item Wlog let $x = \un \un \un..\un s_1...s_it_1..t_j$ where all the $\un$'s and $t_i$'s are sensitive. Let us define inputs $x_i = \un \un \un..\un s_1 \cdots s_it_1. \Bar{t_{i}}\cdots t_j$ for all $i \in [j]$. We know that $x_1, x_2 \cdots x_j$ all correspond to monochromatic cubes (say $j_1$ cubes of color $b$ and $j_2$ of color $\Bar{b}$) and $\hfs(f, x) = \Delta + j$ where $\Delta$ is the dimension of $C(x)$. From the first part of the theorem, we know that there exists a $y \in C(x)$ such that $y$ is of color $b$ and the rest of the resolutions are colored $\Bar{b}$.  Consider the sensitivity of the inputs $y$ and a neighbor $z$ of $y$ in $C(x)$.
            \begin{description}
            \item [$\hfs(y)$] We know that for $y = b_1...b_{\Delta}s_1...s_it_1..t_j$, all the $b_i$'s are sensitive as it is the only resolution of $x$ that has the color $b$. Considering the $t_i$'s note that neighbors on flipping $t_i$ to $\Bar{t_i}$, say $y_i \in C(x_i)$ and hence at least $j_2$ dimensions are sensitive. $\hfs(f, y) \ge \Delta + j_2$
            \item [$\hfs(z)$] Let $z = z_1...z_{\Delta}s_1...s_it_1..t_j$ has an output $\Bar{b}$. Considering the $t_i$'s by a similar argument at least $j_1$ dimensions are sensitive. We have $\hfs(f, z) - 1 \ge j_1$ since $z$ has $y$ as a sensitive neighbor in that direction.
            \end{description}

        From the above we have $$\Delta + j_1 + j_2 = \Delta + j \le \hfs(f, y) + \hfs(f, z) - 1$$
    
        Since $y, z$ are stable inputs we have $\hfs^{(\un)}(f) \le k_1 + k_2 -1$
        \item Suppose for the sake of contradiction, $x$ which does not correspond to a minimum colorful cube gives $\hfs^{(\un)}(f)$. Wlog $x = \un \un \un..\un s_1 \cdots s_it_1. \Bar{t_{i}}\cdots t_j$ where $f(x)= \un$, $t_i$ is a sensitive stable bit, $s_i$ is a non sensitive stable bit. If a $\un$ at $i^{th}$ position is not sensitive consider a $y$ where wlog it is set to $0$. Since it was not sensitive we know that $f(y) = \un$. We will now prove that $\hfs(f, y) \ge \hfs(f, x)$ as follows:
            \begin{description}
                \item Suppose a $x_j = \un$ was sensitive for $x$, we know that setting $x_j$ to some $b$ gives a $x'$ where $f(x') \neq \un$. Note that the same bit is also sensitive for $y$ as changing it to $b$ gives $y'$ such that $C(y') \subset C(x')$ and hence $f(y') \neq \un$
                 \item Suppose $t_j$ was sensitive for $x$ then it has to be that replacing it by $\Bar{t_j}$ gives a $x'$ such that $f(x') \ne \un$, and a similar argument follows.
            \end{description}

        We can recursively keep choose any of the 2 colorful sub cubes of $x$ (in non sensitive dimension) and the $\hfs^{(\un)}$ does not reduce for that input. Hence we can finally get to an input where all $\un$'s are sensitive and that would correspond to a minimum colorful subcube.
    \end{enumerate}

\end{proof}

The lemmas ~\ref{lem:MonoCube} and ~\ref{lem:ColorfulCube} prove the theorem ~\ref{thm:sens}

\subsection{Certificate Complexity in Hazard free Setting}

\begin{theorem}
     Let $f$ be the hazard-free extension of some Boolean function and $k_1, k_2$ are the sizes of the largest prime implicant and implicate respectively, then

     \begin{enumerate}
         \item $\hfcc^{(1)}(f) = k_1, \hfcc^{(0)}(f) = k_2$
         \item $\hfcc^{(\un)}(f) \le k_1 + k_2 -1$
     \end{enumerate}
     
\end{theorem}
\begin{proof}
    \begin{enumerate}
        \item We firstly observe that for any input $x$ giving a stable output $b$, the minimum certificate does not have any unstable bits in it. Note that a subcube corresponding to a certificate (the remaining bits are set to $\un$) is monochromatic, if the minimum certificate contains a $\un$, that can be removed to get $C'$ and any input that matches with $C'$ would be a subcube of the monochromatic subcube and hence $C'$ is a certificate.\\
        Now consider the largest prime implicant $x$, we know from the above observation that $\hfcc^{(1)}$ is at most the number of stable bits in $x$ since any input with a stable output cannot have more stable bits than in $x$. However we know that for $x$ all stable bits are sensitive and hence a certificate should contain all of them. Hence we have $\hfcc^{(1)}(f) = \hfs^{(1)}(f, x) =  k_1$. A similar argument follows for $ \hfcc^{(0)}$

    \item To prove this, we shall use the notion of 'no contradiction' to talk about the intersection of subcubes in the hypercube.  We say $x, y\in \zuon$ do not contradict iff $\nexists i, x_i = b, y_i = \Bar{b}$. We observe that $x, y$ do not contradict iff $C(x) \cap C(y) \ne \phi$. 

    Consider an input $x$ such that $f(x) = \un$, then we know that $C(x)$ has inputs $x_0, x_1$ such that $f(x_0) = 0, f(x_1) = 1$. Suppose $p_0, p_1$ are the inputs in $\zuon$ that represent the maximal monochromatic sub cubes containing $x_0, x_1$, and $D_0, D_1$ are the stable dimensions in $p_0, p_1$ respectively.

    We shall now show that the indices $D_0 \cup D_1$ is a $\un$ certificate for $x$ using the following lemma:

    \begin{lemma}
        Any $y$ such that $y_i = x_i \forall i \in D_0 \cup D_1$ does not contradict with $p_0, p_1$
    \end{lemma}
    Wlog suppose for the sake of contradiction, $y_i = b, {p_0}_i = \Bar{b}$. We know that $i \in D_0$ and hence $y_i = x_i = b$. But we know that $x, p_0$ cannot contradict because $x_0 \in C(x) \cap C(p_0)$. Hence we have a contradiction.

    Hence from the above claim, we know that any $y$ such that $y_i = x_i \forall i \in D_0 \cup D_1$,  $C(y)$ contains an input whose $f$ value is 0 and also an input whose value is 1. We now have,

    $$\hfcc^{(\un)} \le |D_0 \cup D_1| = |D_0| + |D_1| - |D_0 \cap D_1|$$

    We show that $D_0 \cap D_1 \neq \phi$ by contradiction. Assume that $D_0 \cap D_1 \neq \phi$, this would mean that $p_0$ and $p_1$ do not contradict and hence $C(p_0) \cap C(p_1) \ne \phi$ which is a contradiction since they are monochromatic cubes of different colors. Hence we have 

    $$\hfcc^{(\un)} \le  |D_0| + |D_1| - |D_0 \cap D_1| \le k_2 + k_1  - 1$$
    
    \end{enumerate}

\end{proof}




\end{document}